\documentclass[sn-mathphys-ay]{sn-jnl}

\usepackage{graphicx}%
\usepackage{multirow}%
\usepackage{amsmath,amssymb,amsfonts}%
\usepackage{amsthm}%
\usepackage{mathrsfs}%
\usepackage[title]{appendix}%
\usepackage{xcolor}%
\usepackage{textcomp}%
\usepackage{manyfoot}%
\usepackage{booktabs}%
\usepackage{algorithm}%
\usepackage{algorithmicx}%
\usepackage{algpseudocode}%
\usepackage{listings}%
\usepackage{subcaption}

\algrenewcommand\algorithmicrequire{\textbf{Input:}}
\algrenewcommand\algorithmicensure{\textbf{Output:}}
\usepackage[capitalize,nameinlink,noabbrev]{cleveref}

\newcommand{\ket}[1]{ |{#1}  \rangle}
\newcommand{\poly}{\operatorname{poly}}
\DeclareMathOperator*{\argmin}{arg\,min}

\newtheorem{theorem}{Theorem}
\newtheorem{lemma}[theorem]{Lemma}
\newtheorem{fact}[theorem]{Fact}
\newtheorem{corollary}[theorem]{Corollary}
\newtheorem{remark}{Remark}%

\newtheorem{definition}{Definition}%
\Crefname{equation}{Eq.}{Eqs.}

\begin{document}

\title[Article Title]{Hybrid Quantum-Classical Algorithm For Robust Optimization via Stochastic-Gradient Online Learning}

\author[1,2]{\fnm{Debbie} \sur{Lim}}\email{limhueychih@gmail.com}
\equalcont{These authors contributed equally to this work.}

\author[1,3]{\fnm{Joao F.} \sur{Doriguello}}\email{doriguello@renyi.hu}
\equalcont{These authors contributed equally to this work.}

\author[1,4]{\fnm{Patrick} \sur{Rebentrost}}\email{cqtfpr@nus.edu.sg}

\affil[1]{\orgdiv{Centre for Quantum Technologies}, \orgname{National University of Singapore},\\ \orgaddress{\street{3 Science Drive 2}, \city{Singapore}, \postcode{117543}, \country{Singapore}}}

\affil[2]{\orgdiv{Center for Quantum Computing Science}, \orgname{University of Latvia}, \\
\orgaddress{\street{Raina Bulvaris 19}, \city{Riga}, \postcode{LV-1586}, \country{Latvia}}}

\affil[3]{\orgdiv{HUN-REN Alfr\'{e}d R\'{e}nyi Institute of Mathematics},\\
\orgaddress{\street{Re\'{a}ltanoda utca 13-15}, \city{Budapest}, \postcode{053}, \country{Hungary}}}

\affil[4]{\orgdiv{Department of Computer Science}, \orgname{National University of Singapore}, \\
\orgaddress{\street{13 Computing Drive}, \city{Singapore}, \postcode{117417}, \country{Singapore}}}


\abstract{Optimization theory has been widely studied in academia and finds a large variety of applications in industry.
The different optimization models in their discrete and/or continuous settings have catered to a rich source of research problems. Robust convex optimization is a branch of optimization theory in which the variables or parameters involved have a certain level of uncertainty. In this work, we consider the online robust optimization meta-algorithm by Ben-Tal \textit{et al.}\ and show that for a large range of stochastic subgradients, this algorithm has the same guarantee as the original non-stochastic version. We develop a hybrid quantum-classical version of this algorithm and show that an at most quadratic improvement in terms of the dimension can be achieved. The speedup is due to the use of quantum state preparation, quantum norm estimation, and quantum multi-sampling. We apply our quantum meta-algorithm to examples such as robust linear programs and robust semidefinite programs and give applications of these robust optimization problems in finance and engineering. }




\maketitle

\section{Introduction}
\subsection{Robust optimization and online algorithms}
Robust convex optimization is a form of convex optimization with the additional restriction that the optimization data/parameters belong to an uncertainty set and all constraints must hold for every parameter in the set. More precisely, the task is to
\begin{eqnarray}
\nonumber\text{minimize} & & f_0(x)\\
\nonumber\text{subject to} & & f_i(x, u_i)\leq 0, \hspace{1cm}\forall  u_i\in\mathcal{U}, \quad i\in[m] := \{1,\dots,m\},\\
\nonumber & & x\in\mathcal{D},
\end{eqnarray}
where $f_0,f_1,\dots,f_m$ are convex functions in the parameter $x$ and $f_1,\dots,f_m$ are concave in the noise vector $u$. Moreover, $\mathcal{D}\subseteq\mathbb{R}^n$ and the uncertainty set $\mathcal{U}\subseteq\mathbb{R}^d$ are both convex. Such a restrictive setting addresses the natural issue of data inaccuracies in many convex optimization problems, which can lead to substantial fluctuations in the solution. Robust convex optimization was first introduced by \cite{ben1998robust}, which had since then been explored and studied in more depth~\citep{ben2002robust,ben2009robust,bertsimas2011theory}. Despite its theoretical and empirical achievements, the computational cost of adopting robust optimization to large scale problems can become highly prohibitive. In order to tackle this issue, \cite{Ben-Tal2015a} proposed meta algorithms to approximately solve the robust counterpart of a given optimization problem using only an algorithm for the original optimization formulation. Their approach for achieving an efficient oracle-based robust optimization was to reduce the problem to a feasibility question and use a primal-dual technique based on online learning tools such as online (sub)gradient descent~\citep{zinkevich2003online}.

The work of \cite {Ben-Tal2015a} is one of the several recent results that employed online algorithms in the development of learning and optimization algorithms~\citep{bottou1998online, hazan2007logarithmic, mcmahan2010adaptive, shalev2012online, Hazan2016b, Hazan2016a, Fotakis2020, arora2009learning}. Unlike offline algorithms which return a solution based on a complete set of input data, online algorithms produce an output based on partial input data. Feedback in the form of a loss function is received from (a possibly adversarial) nature and the online algorithm uses information from the feedback, together with other pieces of input data, to  update its strategy and outputs a better solution in the next iteration. This process continues in a sequential manner throughout all iterations. There are two quantities often associated with online algorithms: the offline loss and the regret. The offline loss is defined as the minimum loss suffered when choosing the same strategy that minimizes the losses throughout all iterations. On the other hand, the regret is the difference between the offline loss and the loss from some sequence of strategies. Besides taking computational resources such as time and space into account, the design of online algorithms also aims to minimize the regret.  One of the attributes of online algorithms is that their regret bounds hold even for worst-case inputs. 

\subsection{Problem setting}

The robust optimization meta-algorithm of \cite{Ben-Tal2015a} is an online algorithm that uses a subroutine for the original optimization formulation. At each iteration $t$, the noise vectors $u^{(t)}_1,\dots,u^{(t)}_m$ are updated using online (sub)gradient descent as 
\begin{align*}
    u_i^{(t+1)} = \mathcal{P}_{\mathcal{U}} \left(u_i^{(t)} + \eta^{(t)}\nabla_u f_i\big(x^{(t)},u_i^{(t)}\big)\right),
\end{align*}
where $\mathcal{P}_{\mathcal{U}}(u) = \arg\min_{v\in\mathcal{U}}\|u-v\|_2$ is the Euclidean projector operator onto the set $\mathcal{U}$ and $\eta^{(t)}$ is a step size. Once $u_1^{(t+1)},\dots,u_m^{(t+1)}$ are obtained, the authors assume the availability of an optimization oracle $\mathcal{O}_\epsilon$ for the original convex optimization problem that, given an input $(u_1,\dots,u_m)\in\mathcal{U}^m$, outputs either $x\in\mathcal{D}$ that satisfies
\begin{align*}
    f_i(x,u_i) \leq \epsilon, \quad i\in[m],
\end{align*}
or \texttt{INFEASIBLE} if there is no $x\in\mathcal{D}$ for which
\begin{align*}
    f_i(x,u_i) \leq 0, \quad i\in[m].
\end{align*}
The new parameter $x^{(t+1)} = \mathcal{O}_\epsilon(u_1^{(t+1)},\dots,u_m^{(t+1)})$ is then used in the next iteration. \cite{Ben-Tal2015a} proved that such an online meta-algorithm correctly concludes that the robust program is infeasible or outputs an $\epsilon$-approximate solution, i.e., outputs $x\in\mathcal{D}$ that meets each constraint up to $\epsilon$, using $O(D^2 G_2^2/\epsilon^2)$ calls to $\mathcal{O}_\epsilon$, where $D = \max_{u,v\in\mathcal{U}}\|u-v\|_2$ is the diameter of $\mathcal{U}$ and $G_2 \geq \|\nabla_u f_i(x,u)\|_2$, for every $x\in\mathcal{D}$ and $u\in\mathcal{U}$, is an upper bound on the subgradient $\ell_2$-norm.

\subsection{Our results}

In this work, we are interested in providing a quantum speedup for robust optimization in the online and oracle setting of \cite{Ben-Tal2015a}. We stay close to the near-optimal classical framework of Ben-Tal \emph{et al.} to show possible theoretical advantages. The near-optimal setting of this algorithm makes ``dequantization" hard as to our knowledge, no asymptotically better algorithm exists. To this end, we formally specify the input model which allows us to consider the total time complexity of the algorithm. We assume access to a quantum oracle $\mathcal{Q}_\nabla$ for computing the entries of the subgradient $(\nabla_u f_i(x,u_i))_j$ in superposition (\cite{Ben-Tal2015a} assume access to a classical oracle $\mathcal{O}_\nabla$ for computing the subgradients) and a classical oracle $\mathcal{O}_{\mathcal{P}}$ for performing the projection onto $\mathcal{U}$. 
We rephrase and generalize the classical algorithm of \cite{Ben-Tal2015a} in two ways: (i) we encompass the subgradients computation via a stochastic subgradient oracle $\mathcal{O}_g$ that outputs random vectors $g_1,\dots,g_m$ whose expectation values are proportional to the true subgradients $\nabla_u f_1(x,u),\dots,\nabla_u f_m(x,u)$. In their original work, \cite{Ben-Tal2015a} construct the oracle $\mathcal{O}_g$ simply by exactly computing the subgradients entry-wise using $\mathcal{O}_\nabla$; (ii) we use a \emph{stochastic}-(sub)gradient-descent update rule using the random vectors $g_1,\dots,g_m$. Adapting guarantees from online stochastic (sub)gradient descent, we show that the guarantees of the algorithm hold also for a large range of stochastic subgradients. We then propose a hybrid quantum-classical robust optimization meta-algorithm for approximately solving the robust problem that requires less calls to the quantum subgradient oracle $\mathcal{Q}_\nabla$ compared to the number of calls to the classical oracle $\mathcal{O}_{\nabla}$ from \cite{Ben-Tal2015a} if good bounds for the subgradients' $\ell_1$-norm are known\footnote{This happens when, e.g., the constraint functions $f_i(x, u_i)$ are known to be $L$-Lipschitz in $u_i$ $\forall i\in [m]$. In this case, every subgradient has $\ell_1$-norm $L$. Such a situation is common in models like Lasso and ReLU \citep{meng2024lipschitz, geoffrey2020robust, chen2020semialgebraic}. In many applications, the Lipschitz constant is known and can be computed analytically~\citep{avant2023analytical, bhowmick2021lipbab}.}. As a byproduct of such construction, the number of calls to the projection oracle $\mathcal{O}_\mathcal{P}$ is also reduced. The main idea behind our algorithm is to construct an $\ell_1$-sampled stochastic subgradient by measuring the quantum state
\begin{align*}
    \sum_{i=1}^m\sum_{j=1}^d \sqrt{\frac{|(\nabla_u f_i(x,u_i))_j|}{\sum_{k=1}^m \|\nabla_u f_k(x,u_k)\|_1}} |i,j\rangle
\end{align*}
several times. To this end,  we use quantum subroutines for multi-sampling from~\cite{Hamoudi2019,hamoudi2022preparing}.
We summarize our results in \Cref{results_table}.

\begin{table}[t]
\centering
\caption{The number of calls (up to a constant factor) to the subgradient, projection, and optimization oracles $\mathcal{O}_\nabla$/$\mathcal Q_\nabla$, $\mathcal{O}_{\mathcal{P}}$, and $\mathcal{O}_\epsilon$, respectively, in order to obtain an $\epsilon$-approximate solution to a robust convex optimization problem. Here $m$ denotes the number of noise vectors and $d$ is their dimension. In addition, $D \geq \max_{u,v\in\mathcal{U}}\|u-v\|_2$ is an upper bound on the $\ell_2$-diameter of the uncertainty set $\mathcal{U}$. The parameters $G_2,G_1,G_\infty \in\mathbb{R}_{>0}$ are such that $\|\nabla_u f_i(x,u_i)\|_2 \leq G_2$, $\sum_{k=1}^m \|\nabla_u f_k(x,u_k)\|_1 \leq G_1$, and $\max_{k\in[m]}\|\nabla_u f_k(x,u_k)\|_1 \leq G_\infty$ for all $x\in\mathcal{D}$ and $u_1,\dots,u_m\in\mathcal{U}$. Finally, $\delta\in(0,1/2)$ is the failure probability for our algorithm.}\label{results_table}
\renewcommand{\arraystretch}{2}
\begin{tabular}{|c|c|c|c|} 
\hline
Algorithm & Calls to $\mathcal O_\nabla$/$\mathcal Q_\nabla$ & Calls to $\mathcal O_\mathcal{P}$ & Calls to $\mathcal{O}_\epsilon$\\
\hline
\hline
\cite{Ben-Tal2015a} & $\frac{G_2^2 mdD^2}{\epsilon^2}$ & $\frac{G_2^2 mD^2}{\epsilon^2}$ & $\frac{G_2^2D^2}{\epsilon^2}$\\ 
\hline
Hybrid quantum-classical & $\frac{\sqrt{G_1G_\infty} G_2 \sqrt{md}D^2}{\epsilon^2} \log\!\left(\!\frac{DG_2}{\epsilon\delta}\!\right)$ & $\min\{G_1G_\infty, G_2^2 m\}\frac{D^2}{\epsilon^2}$ & $\frac{G_2^2 D^2}{\epsilon^2}$ \\ 
\hline
\end{tabular}
\end{table}

Our final complexity depends on the product $G_1G_\infty$ where $G_1$ and $G_\infty$ are such that $\sum_{k=1}^m \|\nabla_u f_k(x,u_k)\|_1 \leq G_1$ and $\max_{k\in[m]}\|\nabla_u f_k(x,u_k)\|_1 \leq G_\infty$ for all $x\in\mathcal{D}$ and $u_1,\dots,u_m\in\mathcal{U}$. The quantities $G_1,G_\infty$ can be interpreted as bounds on the $\ell_1$- and $\ell_\infty$-norms of the vector $(\|\nabla_u f_1(x,u_1)\|_1,\dots,\|\nabla_u f_m(x,u_m)\|_1)$. It is true that 
\begin{align*}
    \operatorname*{\max}_{i\in[m]}\|\nabla_u f_i(x,u_i)\|_1\sum_{k=1}^m \|\nabla_u f_k(x,u_k)\|_1 \leq m\operatorname*{\max}_{i\in[m]}\|\nabla_u f_i(x,u_i)\|_1^2 \leq md G_2^2
\end{align*}
(recall that $\|\nabla_u f_i(x,u)\|_2 \leq G_2$), therefore one can always take $G_1G_\infty = md G_2^2$. It turns out that if $G_1 G_\infty$ is comparable to $G_2^2$, then we obtain a quadratic improvement in both $m$ and $d$. If $\|\nabla_u f_i(x,u)\|_1 = O(G_2)$ for all $(x,u)\in\mathcal{D}\times\mathcal{U}$ and $i\in I\subseteq[m]$ with $|I|=\Omega(m)$, in which case $G_1 G_\infty = O(m G_2^2)$, then we obtain a quadratic improvement in $d$. On the other hand, if $\|\nabla_u f_i(x,u)\|_1 = O(d G_2)$ for all $(x,u)\in\mathcal{D}\times\mathcal{U}$ and only $i\in I\subseteq[m]$ with $|I| = O(1)$, while $\|\nabla_u f_i(x,u)\|_1 = O(d G_2/m)$ for the remaining $i\in[m]\setminus I$, then we obtain a quadratic improvement in $m$. Finally, if there is no good promise on the quantities $G_1,G_\infty$ and the best one can say is $G_1G_\infty = md G_2^2$, then our hybrid quantum-classical algorithm has the same complexity as the classical algorithm from \cite{Ben-Tal2015a}. There are thus two types of quantum advantage:
\begin{itemize}
    \item Quadratic advantage in $m$: the problem possesses a few constraint functions $f_i(x,u_i)$ that are very sensitive to their corresponding noise vectors $u_i$. The sensitivity here is measured by the subgradients $\nabla_u f_i(x,u)$. Any slight changes to these few noise vectors could cause the corresponding constraint function to change drastically. In other words, the uncertainty given by $\mathcal{U}$ is highly localized, e.g., as in robust portfolio optimization under market stress~\citep{Korn2022} and robust supply chain with a few critical bottlenecks~\citep{bertsimas2004robust};
    \item Quadratic advantage in $d$: most constraint functions are sensitive to only a few entries of their corresponding noise vectors $u_i$, so $\|\nabla_u f_i(x,u)\|_1/\|\nabla_u f_i(x,u)\|_2 = O(1)$. This situation is similar to several other instances in quantum computing where an advantage is observed on sparse inputs \citep{berry2007efficient,harrow2009quantum,berry2014exponential,prakash2014quantum,gilyen2019quantum,Doriguello2025quantumalgorithms}.
\end{itemize}
A problem can exhibit sparsity both in the number of highly sensitive constraint functions and in the number of highly sensitive entries in the noise vectors $u_i$, in which case a quadratic advantage in both $m$ and $d$ is possible as mentioned above. On the other hand, for dense instances, there is no quantum advantage using our hybrid quantum-classical algorithm. By a counting argument, dense instances are much more common, but not necessarily more interesting, as evident in works studying quantum problems with sparse inputs \citep{chen2024low,luo2025space,bellante2025quantum}.

As an application, we show that our quantum meta-algorithm can achieve at most a quadratic improvement in $d$ when solving robust linear programs and robust semidefinite programs. In particular, we consider the case where the uncertainty set takes the form of an ellipsoid~\citep{ben1998robust,ben1999robust,Ben-Tal2015a}. We show that the speedup is dependent on the $\ell_1$, $\ell_\infty$, and Frobenius norm of the matrix that controls uthe shape of the ellipsoid. In addition, we give explicit examples of robust linear programs and robust semidefinite programs in finance and engineering. In finance, the Global Maximum Return Portfolio (GMRP) problem aims to maximize the return of an investment without considering the variance of the market. The GMRP is generalized to a robust linear program since the returns of the assets are not known in practice and therefore are assumed to belong to an uncertainty set. In engineering, the Truss Topological Design (TTD) problem aims to optimize the choice of geometry, topology, and size of a truss structure in order to withstand external loads. This problem is expressed mathematically as a robust semidefinite program since information on the external loads is not perfectly known and hence is assumed to belong to an uncertainty set. Our quantum meta-algorithm solves these problems with a total runtime that is dependent on the Frobenius, $\ell_2$, and $\ell_\infty$-norms of the matrices that control the shape of the ellipsoid.

\subsection{Related work}

Online convex optimization seeks a solution from a convex set that optimizes a convex objective function in an online manner~\citep{shalev2012online, Hazan2016b}. The commonly used convex optimization techniques include variants of (sub)gradient descent, mirror descent, coordinate descent~\citep{zinkevich2003online, Hazan2016b, Flaxman2004, Hazan2007, Wang2014}, the multiplicative weight update method by~\cite{Arora2016a}, and Newton's method~\citep{Schraudolph2007}. Applications of these techniques can be found in~\cite{Arora2016, Helmbold2009, Helmbold1998}. Online convex optimization has wide applications in various fields such as telecommunication~\citep{belmega2022online}, aviation~\citep{patel2011trajectory, liu2017survey}, and finance~\citep{li2014online, kim2016online}, 

\cite{Freund1997} showed that the multiplicative weight update method can be applied in the general decision-theoretic setting of online learning, and thus in a more general class of learning problems. The same technique was also used to develop a boosting algorithm that does not require the performance of the weak learning algorithm to be known \textit{a priori}. The goal of such boosting algorithms is to convert a weak and inaccurate learner into a stronger and more accurate one. \cite{Ben-Tal2015a} provided two online meta-algorithms for solving a robust optimization problem: the first algorithm, which is the main algorithm of this work, approximately solves the problem using the dual subgradient method, assuming that noise vectors belong to a convex uncertainty set. Their second algorithm removes the convexity assumption on the uncertainty set but assumes some linearity properties on the constraint functions, and uses a “pessimization oracle” to find the worst-case noise assignment. The survey by \cite{bertsimas2011theory} compiled the results of robust linear, quadratic, semidefinite, and discrete optimization on different uncertainty sets, such as ellipsoidal, polyhedral, cardinality constraint, and norm uncertainties.

In the quantum setting, the runtime of a number of classical online optimization algorithms have been improved quadratically. The quantum AdaBoost algorithm by 
\cite{arunachalam2020quantum} achieves a quadratic speedup in the Vapnik-Chervonenkis (VC) dimension of the weak learner's hypothesis class as compared to the classical AdaBoost by \cite{Freund1997}. The quantum AdaBoost is capable of achieving such an objective in a short time using quantum techniques such as amplitude amplification and estimation~\citep{Brassard2002} and quantum mean estimation~\citep{brassard2011optimal, nayak1999quantum}. Subsequently, \cite{izdebski2020improved} provided a substantially faster and simpler quantum boosting algorithm based on the SmoothBoost algorithm by \cite{servedio2003smooth}. In the context of zero-sum games, \cite{van2019quantum} showed how to efficiently sample from the Gibbs distribution, leading to a quadratic speedup compared to its classical counterpart~\citep{Grigoriadis1995} in terms of the dimension of the payoff matrix. Furthermore, for the task of classifying $n$ $d$-dimensional data points, \cite{li2019sublinear} proposed a sublinear time quantum algorithm quadratically better than the classical algorithm by \cite{clarkson2012sublinear}. Quantum SDP solvers were developed using the Arora-Kale method~\citep{Brandao2017,VanApeldoorn2020}. 

\section{Preliminaries}
For $n\in\mathbb{N}$, let $[n]:=\{1, \dots, n\}$. If a vector is time dependent, we denote its time dependency as a superscript. Given a vector $v\in\mathbb{R}^n$ and $p\in[1,\infty]$, its $\ell_p$-norm is $\| v\|_p := (\sum_{i=1}^n \vert v_i\vert^p)^{1/p}$. Given a matrix $A\in\mathbb{R}^{n\times m}$ and $p\in[1,\infty]$, define $\|A\|_p := \sup_{x\in\mathbb{R}^m:\|x\|_p =1}\|Ax\|_p$. In the special cases $p=1,2,\infty$, one have $\|A\|_1 = \max_{j\in[m]}\sum_{i=1}^n |A_{ij}|$, $\|A\|_\infty = \max_{i\in[n]}\sum_{j=1}^m |A_{ij}|$, and $\|A\|_2 = \sigma_{\max}(A)$, where $\sigma_{\max}(A)$ is the maximum singular value of $A$. Let also $\|A\|_F := \sqrt{\sum_{i=1}^n\sum_{j=1}^m |A_{ij}|^2}$ be its Frobenious norm.  Given a convex function $f:\mathcal{D}\to\mathbb{R}$, its subgradient at $x$ is any vector $\nabla f(x)$ such that $f(y) \geq f(x) + \nabla f(x)^\top (y-x)$ for all $y\in\mathcal{D}$. Similarly for a concave function $f:\mathcal{D}\to\mathbb{R}$, its subgradient at $x$ is any vector $\nabla f(x)$ such that $f(y) \leq f(x) + \nabla f(x)^\top (y-x)$ for all $y\in\mathcal{D}$. We use $\bar{0}$ to denote the all zeros vector and $\ket{\bar{0}}$ to denote the state $\ket{0}\otimes\cdots\otimes\ket{0}$, where the number of qubits is clear from the context. We use $\widetilde O(\cdot)$ to hide the polylog factor, i.e., $\widetilde O(f(n)) = O(f(n)\cdot \poly\log(f(n)))$.

\subsection{The robust convex optimization framework}

Consider a convex optimization problem with an additional parameter: 
\begin{eqnarray}
\nonumber\text{minimize} & & f_0(x)\\
\nonumber\text{subject to} & & f_i(x, u_i)\leq 0, \qquad i\in[m],\\
\nonumber& & x\in\mathcal{D},
\end{eqnarray}
where $\mathcal{D}\subseteq\mathbb{R}^n$ is a convex set, $f_0,f_1,\dots,f_m$ are convex functions on the parameter $x\in\mathcal{D}$, and $u_1,\dots,u_m\in\mathbb{R}^d$ are fixed parameters. The robust formulation of the above convex optimization problem is
\begin{align}
    \label{robust_convex_opt}
    \begin{aligned}
    \text{minimize} & \quad f_0(x)\\
    \text{subject to} & \quad f_i(x, u_i)\leq 0, \qquad\forall  u_i\in\mathcal{U}, \quad i\in[m], \\
     & \quad x\in\mathcal{D},
    \end{aligned}
\end{align}
where the noise vectors $u_1,\dots,u_m$ now belong to an uncertainty convex set $\mathcal{U}\subseteq\mathbb{R}^d$ and we further assume that $f_1,\dots,f_m$ are concave on the noise vectors.  For convenience we summarize our assumptions for the rest of the work. For all $i\in[m]$,
\begin{enumerate}
    \item $\forall u\in\mathcal{U}$, $f_i(x,u)$ is convex in $x$;
    \item $\forall x\in\mathcal{D}$, $f_i(x,u)$ is concave in $u$;
    \item $\mathcal{D}\subseteq\mathbb{R}^n$ and $\mathcal{U}\subseteq\mathbb{R}^d$ are convex;
    \item $D \geq \max_{u,v\in\mathcal{U}}\| u-v\|_2$ is a bound on the $\ell_2$-diameter of $\mathcal{U}$;
    \item $F \geq \max_{x\in\mathcal{D}, u\in\mathcal{U}}|f_i(x, u)|$ is a bound on the maximum value of the constraint functions;
    \item $G_2 \geq \max_{x\in\mathcal{D}, u\in\mathcal{U}}\|\nabla_u f_i(x,u)\|_2$ is a bound on the $\ell_2$-norm of the subgradients of the constraint functions;
    \item $G_1 \geq \sum_{k=1}^m \|\nabla_u f_k(x,u_k)\|_1$ for all $x\in\mathcal{D}$ and $u_1,\dots,u_m\in\mathcal{U}$;
    \item $G_\infty \geq \max_{k\in[m]}\|\nabla_u f_k(x,u_k)\|_1$ for all $x\in\mathcal{D}$ and $u_1,\dots,u_m\in\mathcal{U}$.
\end{enumerate}

It is known that one can turn an optimization problem into a feasibility problem by using binary search over its optimal values~\citep{arora2012multiplicative}. For simplicity, we shift the objective function by $z$, the current guess of the optimal value,  and let the first constraint be $f_0(x, u)\leq 0$. Hence, a robust optimization problem can be reduced to the following feasibility problem:
\begin{align}
    \label{feasibility_problem}
    \exists ? x\in\mathcal{D}:\quad f_i(x, u_i)\leq 0, \quad \forall u_i\in\mathcal{U}, \quad i\in[m].
\end{align}
We say that $x\in\mathcal{D}$ is an $\epsilon$-approximate solution to the above problem if $f_i(x,u_i) \leq \epsilon$ for all $u_1,\dots,u_m\in\mathcal{U}$ and $i\in[m]$.

\subsection{The computational model and oracles}\label{computational_model}

We use the quantum circuit model~\citep{nielsen2010quantum} whereby an application of a quantum gate is equivalent to performing an elementary operation. In both classical and quantum settings,  we assume an arithmetic model that ignores errors from fixed-point representation and each elementary operation takes constant time. The query complexity of a classical/quantum algorithm with some input length is the maximum number of queries the algorithm makes on any input. The time complexity of a classical/quantum computation is in terms of the number of basic gates performed. Given states $\ket{x_1}, \dots, \ket{x_n}$, where $x_i\in[0, 1]$ for $i\in[n]$, we assume that we can do, in $O(1)$ time, a controlled-rotation on a superposition of  $\ket{x_1}, \dots, \ket{x_n}$:
\begin{align*}
    \ket{x_i}\ket{0}\mapsto \ket{x_i}\left(\sqrt{ x_i}\ket{0}+\sqrt{1- x_i}\ket{1}\right), \quad \forall i\in[n].
\end{align*}

We assume access, both classical and quantum, to a \emph{noise memory} that stores and allows updates on vectors (which shall commonly be the noise vectors $u_1,\dots,u_m\in\mathcal{U}$).
\begin{definition}[Noise memory]\label{noise_memory}
Let $u_1,\dots,u_m\in\mathbb{R}^d$ be a sequence of $m$ vectors that is stored in a classical memory of size $O(md)$. We say that we have access to the noise memory if we can query and update a single entry of any vector in $O(1)$ time. In the quantum setting, we say that we have quantum access to the noise memory if there exists a unitary operator $U$ such that queries can be made in superposition in $O(1)$ time, i.e.,
\begin{align*}
    U: \ket{i}\ket{j}\ket{\bar{0}}\rightarrow 
\ket{i}\ket{j}\ket{(u_i)_{j}}, \quad \forall i\in[m],j\in[d],
\end{align*}
and updates take $O(1)$ time, as in the classical memory. 
The third  register is assumed to contain sufficiently many qubits to allow accurate computation of all subsequent computations. 
\end{definition}

Our abstract model of noise memory model can be broken down into more standard models, e.g., the classical noise memory can be identified as a random-access memory (RAM), whose entries can be queried in constant time. The quantum noise memory, on the other hand, can be more commonly described as a quantum-readable classical-writable classical memory (QRAM)~\citep{giovannetti2008architectures,giovannetti2008quantum,allcock2023constant}, which is accessed in superposition by the quantum computer. Even though the $O(1)$ access time might be too strong of an assumption compared to traditional RAM or QRAM models wherein a query requires $O(\poly\log{md})$ time, such an overhead would only be polylogarithmic and so we opt to omit it for the sake of clarity. We mention that most QRAM circuits have T-gate count of $O(md)$ and T-gate depth of $O(\log{md})$~\citep{giovannetti2008architectures,giovannetti2008quantum,Arunachalam2015robustness,babbush2018encoding}, with new proposals achieving a T-gate count of $O(\sqrt{md})$ and a T-gate depth of $O(\log\log(md))$~\cite{Mukhopadhyay2025quantum}.

As mentioned in the introduction and following \cite{Ben-Tal2015a}, we assume access to a subgradient oracle to access an entry of the subgradient of the constraint functions, a projection oracle to update a noise vector in the noise memory by projecting it back onto the set $\mathcal{U}$, and an optimization oracle to solve the original non-robust convex optimization problem. 
\begin{definition}[Subgradient oracle $\mathcal O_\nabla$/$\mathcal Q_\nabla$]\label{subgradient_oracle}
Let $f_i:\mathcal{D}\times\mathcal{U}\rightarrow\mathbb{R}$ be a function for $i\in[m]$. Assume the noise vectors $u_1, \dots, u_m\in\mathcal{U}$ are stored in the noise memory. We say that we have classical access to a subgradient oracle $\mathcal{O}_\nabla$ if there is an operation that, given $i\in[m], j\in[d]$, and $x\in\mathcal{D}$, returns the subgradient $\left(\nabla_u f_i(x, u_i)\right)_j$. We say we have quantum access to a subgradient oracle $\mathcal{Q}_\nabla$ if there is a unitary $U_{\nabla_u f}$ such that, given $x\in\mathcal{D}$, 
\begin{align*}
    U_{\nabla_u f} : |i\rangle|j\rangle|\bar{0}\rangle \mapsto |i\rangle|j\rangle\ket{ (\nabla_u f_i(x,u_i))_j}, \quad \forall i\in[m],j\in[d].
\end{align*}
%
The time complexity of the classical/quantum oracle is denoted by $\mathcal{T}_\nabla$.
\end{definition}
\begin{definition}[Projection oracle $\mathcal{O}_\mathcal{P}$]
Assume the vectors $u_1,\dots,u_m\in\mathbb{R}^d$ are stored in the noise memory. We say that we have access to a projection onto $\mathcal{U}\subseteq\mathbb{R}^d$ if we have access to an operation $\mathcal{P}_{\mathcal{U}}$ in time $\mathcal{T}_{\mathcal P}$ that updates $u_i$ to $\mathcal{P}_{\mathcal{U}}(u_i)$, $i\in[m]$, where
\begin{align*}
    \mathcal{P}_{\mathcal{U}}(u_i) = \argmin_{q\in\mathcal{U}}\left\| q - u_i\right\|_2.
\end{align*}
\end{definition}
\begin{definition}[Optimization oracle $\mathcal{O}_\epsilon$]\label{classical_opt_oracle}
Let $f_i:\mathcal{D}\times\mathcal{U}\rightarrow \mathbb{R}$ be a function for $i\in[m]$. Furthermore, let the noise vectors $u_1,\dots,u_m\in\mathcal{U}$ be stored in the noise memory. The optimization oracle $\mathcal{O}_\epsilon$ either returns a vector $x\in\mathcal{D}$ such that, for all $i\in[m]$, 
 \begin{align*}
    f_i(x, u_i)\leq\epsilon,
 \end{align*}
or outputs {\rm \texttt{INFEASIBLE}} if there does not exist a vector $x\in\mathcal{D}$ such that, for all $i\in[m]$,
\begin{align*}
    f_i(x, u_i)\leq 0.
\end{align*}
The time complexity of the oracle is denoted by $\mathcal{T}_\epsilon$.
\end{definition}
We note that the projection operator $\mathcal P_{\mathcal U}$ operates the same on the noise memory regardless of classical or quantum access to it.

\subsection{Online stochastic subgradient descent}

One of the simplest, and yet most versatile algorithms for online convex optimization is the online (sub)gradient descent introduced by \cite{zinkevich2003online}. At each iteration the algorithm takes a step from the current point in the direction of the gradient of the function from the previous iteration. Since the step can take the point out of the underlying convex set, it must be projected back to the convex set. Here we shall use its stochastic subgradient variant\footnote{We shall focus on the problem of online maximization of \emph{concave} functions rather than online minimization of convex functions. Both formulations are equivalent.}, where we are given access to a random vector whose expectation is the subgradient at a given point (up to a factor close to $1$). More specifically, for a fixed $\nu\in[0,1/2]$ and given concave (loss) functions $h^{(t)}:\mathcal{U}\subseteq\mathbb{R}^d\rightarrow\mathbb{R}$, $t\in[T]$, assume we are given access to noisy subgradient oracles $ {\mathcal{O}}_g^{(t)}$ which for all $u\in \mathcal U$ output $g^{(t)}_u\in\mathbb{R}^d$ such that
\begin{align}
    \label{eq:noisy_gradient}
    \mathbb{E}[g^{(t)}_u ] = \lambda^{(t)} \nabla h^{(t)}(u) \quad \text{where}~ |\lambda^{(t)} - 1| \leq \nu, \quad\text{and}\quad \mathbb{E}[\|g_u^{(t)} \|^2_2]\leq \widetilde G_2^2. 
\end{align}
We stress that the parameters $\lambda^{(t)}$ do not need to be known. The slightly modified online stochastic subgradient descent algorithm is described in \Cref{online_stochastic_gradient_descent}. 

We now prove a regret guarantee for \Cref{online_stochastic_gradient_descent}. The functions $h^{(t)}$ will play the role as loss functions. The proof is inspired by~\citep[Theorems~3.1 \&~3.4]{Hazan2016b}.
\begin{algorithm}
\caption{Online stochastic subgradient descent}
\label{online_stochastic_gradient_descent}
\begin{algorithmic}[1]
    \Require Convex set $\mathcal{U}$, $u_1\in\mathcal{U}$, parameters $D,\widetilde G_2$ and steps sizes $\eta^{(t)} = \frac{D}{(1-\nu)\widetilde G_2\sqrt{t}}$;
    \For {$t=1$ to $T$} 
    \State Let $g^{(t)} = \mathcal{O}_g^{(t)}(u^{(t)})$;
    \State $u^{(t+1)}\gets\mathcal{P}_{\mathcal{U}}(u^{(t)} + \eta^{(t)} g^{(t)})$;
    \EndFor
\end{algorithmic}
\end{algorithm}
\begin{lemma}
    \label{lem:regret_stochastic_gradient}
    Let $\nu\in[0, 1/2]$. Online stochastic subgradient descent from {\rm \Cref{online_stochastic_gradient_descent}} with noisy subgradient $\mathcal{O}_g^{(t)}$ satisfying {\rm \Cref{eq:noisy_gradient}} and with step size $\eta^{(t)} = \frac{D}{(1-\nu)\widetilde G_2\sqrt{t}}$ guarantees that
    \begin{align*}
        \max_{u^\ast\in\mathcal{U}}\frac{1}{T} \sum_{t=1}^T h^{(t)} (u^\ast) - \frac{1}{T}\sum_{t=1}^T \mathbb{E}[h^{(t)}(u^{(t)})] \leq \frac{3(1+4\nu)\widetilde G_2D}{2\sqrt{T}}.
    \end{align*}
\end{lemma}
\begin{proof}
    Let $u^\ast \in \arg\max_{u\in\mathcal{U}} \sum_{t=1}^T h^{(t)}(u)$. Observe that, by concavity,
    \begin{align*}
        h^{(t)}(u^\ast) - h^{(t)}(u^{(t)}) \leq \nabla h^{(t)}(u^{(t)})\cdot (u^\ast - u^{(t)}).
    \end{align*}
    On the other hand,
    \begin{align*}
        \|u^\ast - u^{(t+1)}\|^2_2 &= \|u^\ast - \mathcal{P}_{\mathcal{D}}(u^{(t)} + \eta^{(t)}g^{(t)})\|^2_2 \\
        &\leq \|u^\ast - u^{(t)} - \eta^{(t)}g^{(t)}\|^2_2 \\
        &= \|u^\ast - u^{(t)}\|_2^2 + \eta_t^2\|g^{(t)}\|_2^2 - 2\eta^{(t)}g^{(t)}\cdot(u^\ast - u^{(t)}),
    \end{align*}
    and therefore
    \begin{align*}
        2g^{(t)}\cdot(u^\ast - u^{(t)}) \leq \frac{\|u^\ast - u^{(t)}\|^2_2 - \|u^\ast - u^{(t+1)}\|_2^2}{\eta^{(t)}} + \eta^{(t)} \|g^{(t)}\|_2^2.
    \end{align*}
    By noticing that $\lambda^{(t)}\mathbb{E}[\nabla h^{(t)}(u^{(t)})\cdot(u^\ast - u^{(t)})] = \mathbb{E}[g^{(t)}\cdot(u^\ast - u^{(t)})]$, and setting $\eta^{(t)} = \frac{D}{(1-\nu)\widetilde{G}_2\sqrt{t}}$ (with $1/\eta^{(0)} := 0$), we obtain
    \begin{align*}
        2\mathbb{E}\left[\sum_{t=1}^T h^{(t)}(u^\ast) - h^{(t)}(u^{(t)}) \right]
        &\leq 2\mathbb{E}\left[\sum_{t=1}^T \nabla h^{(t)}(u^{(t)})\cdot (u^\ast - u^{(t)}) \right]\\
        &= 2\mathbb{E}\left[\sum_{t=1}^T \frac{1}{\lambda^{(t)}} g^{(t)}\cdot (u^\ast - u^{(t)}) \right] \\
        &\leq \mathbb{E}\left[\sum_{t=1}^T \frac{\|u^\ast - u^{(t)}\|^2_2 - \|u^\ast - u^{(t+1)}\|_2^2}{\lambda^{(t)}\eta^{(t)}} \right] + \widetilde G_2^2\sum_{t=1}^T \frac{\eta^{(t)}}{\lambda^{(t)}}\\
        &\leq \mathbb{E}\left[\sum_{t=1}^T \frac{\|u^\ast - u^{(t)}\|^2_2}{\lambda^{(t)}\eta^{(t)}} - \frac{\|u^\ast - u^{(t)}\|^2_2}{\lambda^{(t-1)}\eta^{(t-1)}}\right] + \widetilde G_2^2\sum_{t=1}^T \frac{\eta^{(t)}}{\lambda^{(t)}}\\
        &\leq D^2\sum_{t=1}^T \left(\frac{1}{\lambda^{(t)}\eta^{(t)}} - \frac{1}{\lambda^{(t-1)}\eta^{(t-1)}}\right) + \widetilde G_2^2\sum_{t=1}^T \frac{\eta^{(t)}}{\lambda^{(t)}}\\
        &\leq \frac{D^2}{\lambda^{(T)}\eta^{(T)}} + \widetilde G_2^2\sum_{t=1}^T \frac{\eta^{(t)}}{\lambda^{(t)}}\\
        &\leq D\widetilde G_2\sqrt{T}\left(1 + \frac{2}{(1-\nu)^2} \right)\\
        &\leq (3+12\nu)D\widetilde G_2\sqrt{{T}},
    \end{align*}
    where the second last inequality follows from $\sum_{t=1}^t t^{-1/2} \leq 2\sqrt{T}$ and $\lambda^{(t)} \geq 1-\nu$, and the last one from $\frac{1}{(1-\nu)^2} \leq (1 + 2\nu)^2 \leq 1 + 6\nu$. 
\end{proof}
\begin{remark}
    In online stochastic gradient descent~{\rm \citep[Theorem~3.4]{Hazan2016b}}, it is usually assumed  access to an oracle that, for all $u\in\mathcal U$, returns $g_u^{(t)}\in\mathbb{R}^d$ such that
    \begin{align*}
        \mathbb{E}[g^{(t)}_u ] = \nabla h^{(t)}(u) \quad\text{and}\quad \mathbb{E}[\|g_u^{(t)} \|^2_2]\leq G^2. 
    \end{align*}
    {\rm \Cref{online_stochastic_gradient_descent}} with access to such an oracle achieves a regret upper-bounded by $\frac{3GD}{2\sqrt T}$. Notice that the bound derived in {\rm \Cref{lem:regret_stochastic_gradient}} has an extra factor of $1+4\nu$ due to the noisy subgradient in the first equality of {\rm \Cref{eq:noisy_gradient}}, where the $\lambda$ parameter is bounded away from $1$ by at most $\nu$. 
\end{remark}

\section{A stochastic dual-subgradient meta-algorithm}

In this section, we generalize the online oracle-based meta-algorithm of \cite{Ben-Tal2015a} used to solve \Cref{robust_convex_opt}. Instead of exactly computing the subgradients $\nabla_u f_i(x,u)$ as originally done by \cite{Ben-Tal2015a}, we now have access to a stochastic subgradient oracle $\mathcal{O}_g$ which, on input $(x,u_1,\dots,u_m)$, outputs $g_1,\dots,g_m\in\mathbb{R}^d$ that, for a given $\nu\in[0,1/2]$, satisfy
\begin{align}\label{eq:gradient_oracle}
    \forall i\in[m]: \quad \mathbb{E}[g_i] = \lambda\nabla_u f_i(x,u_i) \quad\text{where}~|\lambda - 1|\leq \nu, \quad\text{and}\quad \mathbb{E}[\|g_i\|_2^2] \leq \widetilde{G}_2^2.
\end{align}
We stress that $\lambda$ does not need to be fixed or known and can change at each call of the oracle $\mathcal{O}_g$.
The algorithm is shown in \Cref{algo}. The following theorem shows that \Cref{algo} solves the robust optimization problem in \Cref{feasibility_problem}.

\begin{algorithm}[h]
\caption{Online dual-subgradient robust optimization algorithm with subgradient oracle}
\label{algo}
\begin{algorithmic}[1]
    \Require Convex sets $\mathcal{D},\mathcal{U}$, target accuracy $\epsilon>0$, parameters $\nu,D,F,\widetilde{G}_2$, oracle $\mathcal O_g$;
    \State Set $T = \big\lceil\frac{1}{\epsilon^2} \max\big\{4F\log(\frac{m}{\delta}), \frac{9}{4}(1+4\nu)^2D^2 \widetilde{G}_2^2\big\} \big\rceil$ and $\eta^{(t)}=\frac{D}{(1-\nu)\widetilde{G}_2\sqrt{t+1}}$;
    \State Initialize $(u^{(0)}_1, \dots, u^{(0)}_m)\in\mathcal{U}^m$ arbitrarily in the noise memory and $x^{(0)}\in\mathcal{D}$;
    \For {$t=0$ to $T-1$}
    \State $g_{1}^{(t)},\dots,g_m^{(t)}\gets \mathcal O_g(x^{(t)},u^{(t)}_1,\dots,u^{(t)}_m)$;
    \For {$i=1$ to $m$}
    \State $u_i^{(t+1)}\gets\mathcal{P}_{\mathcal{U}}(u_i^{(t)} + \eta^{(t)} g_i^{(t)})$;
    \EndFor
    \State $x^{(t+1)} \gets\mathcal{O}_\epsilon(u_1^{(t+1)},\dots,u_m^{(t+1)})$;
    \If{oracle declares infeasibility}
    \Return \texttt{INFEASIBLE};
    \EndIf
    \EndFor
    \Ensure $\bar{x}=\frac{1}{T}\sum_{t=1}^T x^{(t)}$;
\end{algorithmic}
\end{algorithm}

\begin{theorem}\label{thr:correctness_general}
    For $i\in[m]$, let $f_i: \mathcal{D}\times \mathcal{U}\to \mathbb{R}$ be a convex function over variables from a convex set $\mathcal{D}\subseteq\mathbb{R}^n$ and parameters from a convex set $\mathcal{U}\subseteq\mathbb{R}^d$ with diameter $\max_{u,v\in\mathcal{U}}\| u-v\|_2\leq D$. Let $F\in\mathbb{R}_{>0}$ such that $|f_i(x,u)| \leq F$ for all $x\in\mathcal{D}$ and $u\in\mathcal{U}$. Assume access to a stochastic subgradient oracle $\mathcal{O}_g$ which, on input $(x,u_1,\dots,u_m)$, outputs $g_1,\dots,g_m\in\mathbb{R}^d$ that satisfy {\rm \Cref{eq:gradient_oracle}} for $\nu\in[0,1/2]$ and $\widetilde{G}_2\in\mathbb{R}_{>0}$. Let $T = \big\lceil\frac{1}{\epsilon^2} \max\big\{4F\log(\frac{m}{\delta}), \frac{9}{4}(1+4\nu)^2D^2 \widetilde{G}_2^2\big\} \big\rceil$. {\rm \Cref{algo}} either returns an $3\epsilon$-approximate solution to the robust program {\rm \Cref{feasibility_problem}} with probability at least $1-\delta$ or successfully asserts that it is infeasible.
\end{theorem}
\begin{proof}
    Suppose first that the algorithm returns \texttt{INFEASIBLE}. This means that for some $t\in[T]$ there is no $x\in\mathbb{R}^n$ such that
    \begin{align*}
        f_i(x, u^{(t)}_{i})\leq \epsilon, \quad i\in[m].
    \end{align*}
    This implies the robust counterpart \Cref{feasibility_problem} cannot be feasible, as there is an admissible perturbation that makes the original problem infeasible.
        
    Suppose now a solution $\bar{x} = \frac{1}{T}\sum_{i=1}^t x^{(t)}$ is obtained. This means that $f_i(x^{(t)}, u^{(t)}_{i})\leq \epsilon$ for all $t\in[T]$ and $i\in[m]$, thus
    \begin{align}
        \label{eq:bound2-1}
        \forall i\in[m]: \quad \frac{1}{T}\sum_{t=1}^T f_i(x^{(t)}, u^{(t)}_{i}) \leq \epsilon.
    \end{align}
     Moreover, for fixed $i\in[m]$, note that the random variables $Z^{(t)}_i := f_i(x^{(t)}, u_i^{(t)}) - \mathbb{E}_t\big[f_i(x^{(t)}, u_i^{(t)})\big]$ for $t\in[T]$ form a martingale differences sequence with respect to internal randomness of the oracle $\mathcal{O}_g$, where $\mathbb{E}_t[\cdot]$
     is conditioned on the randomness of previous steps $[t-1]$,
     and
    \begin{align*}
        \big\vert Z_i^{(t)}\big\vert \leq \big\vert f_i(x^{(t)}, u_i^{(t)})\big\vert + \mathbb{E}_t\big[\big\vert f_i(x^{(t)}, u_i^{(t)})\big\vert\big] \leq 2F.
    \end{align*}
    Therefore, by Azuma's inequality~\citep[Lemma~A.7]{cesa2006prediction},
    \begin{align}
        \label{eq:bound2-3}
        \operatorname{Pr}\left[\frac{1}{T}\sum_{t=1}^T\mathbb{E}_t\big[f_i(x^{(t)}, u_i^{(t)})\big] - \frac{1}{T}\sum_{t=1}^T f_i(x^{(t)}, u_i^{(t)}) \geq 2F\sqrt{\frac{\log(m/\delta)}{T}} \right] \leq \frac{\delta}{m}.
    \end{align}
    On the other hand, we can use the regret bound from \Cref{lem:regret_stochastic_gradient} with $h^{(t)}(u_i^{(t)}) = f_i(x^{(t)}, u_i^{(t)})$ for each $i\in[m]$. Then
    \begin{align}
        \label{eq:bound2-2}
        \forall i\in[m]: \quad \max_{u^\ast\in\mathcal{U}}\frac{1}{T}\sum_{t=1}^T f_i(x^{(t)}, u^{\ast}) - \frac{1}{T}\sum_{t=1}^T \mathbb{E}_t[f_i(x^{(t)}, u_i^{(t)})] \leq \frac{3(1+4\nu)D\widetilde{G}_2}{2\sqrt{T}}.
    \end{align}
    Combining \Cref{eq:bound2-1,eq:bound2-3,eq:bound2-2} leads to, for all $i\in[m]$,
    \begin{align}
        \epsilon \geq \frac{1}{T}\sum_{t=1}^T f_i(x^{(t)}, u^{(t)}_{i}) &\geq \frac{1}{T}\sum_{t=1}^T \mathbb{E}_t[f_i(x^{(t)}, u^{(t)}_{i})] - 2F\sqrt{\frac{\log(m/\delta)}{T}} \nonumber\\
        &\geq \max_{u^\ast\in\mathcal{U}}\frac{1}{T}\sum_{t=1}^T f_i(x^{(t)}, u^{\ast}) - 2F\sqrt{\frac{\log(m/\delta)}{T}} - \frac{3(1+4\nu)D\widetilde{G}_2}{2\sqrt{T}}\nonumber\\
        &\geq \max_{u^\ast\in\mathcal{U}}f_i(\bar{x}, u^{\ast}) - 2F\sqrt{\frac{\log(m/\delta)}{T}} - \frac{3(1+4\nu)D\widetilde{G}_2}{2\sqrt{T}}\label{eq:chain_bounds}
    \end{align}
    with probability at least $1-\delta$ by a union bound, where the final inequality comes from the convexity of $f_i$ with respect to $x$. Using $T = \big\lceil\frac{1}{\epsilon^2} \max\big\{4F\log(\frac{m}{\delta}), \frac{9}{4}(1+4\nu)^2 D^2 \widetilde{G}_2^2\big\} \big\rceil$, we arrive, with probability at least $1-\delta$, at
    \[
        f_i(\bar{x},u_i) \leq 3\epsilon, \quad\quad \forall u_i\in\mathcal{U}, \quad i\in[m]. \qedhere
   \]
\end{proof}
\begin{remark}\label{rem:remark1}
    For $i\in[m]$, suppose the function $f_i(x,u)$ is linear in $u$ for all $x\in\mathcal{D}$, i.e., $f_i(x,u) = \alpha_i(x)^\top u + \beta_i(x)$ for functions $\alpha_i:\mathbb{R}^n\to\mathbb{R}^d$ and $\beta_i:\mathbb{R}^n\to\mathbb{R}$. Then {\rm \Cref{eq:bound2-3,eq:bound2-2}} can be applied directly to $\alpha_i(x)^\top u$, in which case
    \begin{align*}
        \frac{1}{T}\sum_{t=1}^T \alpha_i(x^{(t)})^\top u^{(t)}_{i} &\geq \frac{1}{T}\sum_{t=1}^T \mathbb{E}_t[\alpha_i(x^{(t)})^\top u^{(t)}_{i}] - 2F\sqrt{\frac{\log(m/\delta)}{T}} \\
        &\geq \max_{u^\ast\in\mathcal{U}}\frac{1}{T}\sum_{t=1}^T \alpha_i(x^{(t)})^\top u^{\ast} - 2F\sqrt{\frac{\log(m/\delta)}{T}} - \frac{3(1+4\nu)D\widetilde{G}_2}{2\sqrt{T}}.
    \end{align*}
    {\rm \Cref{eq:chain_bounds}} can be obtained by summing $\frac{1}{T}\sum_{t=1}^T \beta_i(x^{(t)})$ to both sides of the equation above. The advantage of doing so is that $F$ becomes an upper bound on $|\alpha_i(x)^\top u|$ instead of $|f_i(x,u)|$.
\end{remark}

As previously mentioned, \cite{Ben-Tal2015a} originally used an oracle $\mathcal{O}_g$ that, on input $(x,u_1,\dots,u_m)$, outputs $\nabla_u f_1(x,u_1),\dots,\nabla_u f_m(x,u_m)$, i.e., the exact subgradients. This was done simply by calling the subgradient oracle $\mathcal{O}_\nabla$ from \Cref{subgradient_oracle} in order to compute each entry $(\nabla_u f_i(x,u_i))_j$, for $i\in[m],j\in[d]$. The overall complexity of their oracle $\mathcal{O}_g$ is then $md\mathcal{T}_\nabla$. In the next section, we show that a hybrid quantum-classical algorithm can reduce the overall complexity of $\mathcal{O}_g$ with regard to the subgradient oracle. 

\section{Hybrid quantum-classical online robust optimization}

In this section, we present our hybrid quantum-classical online algorithm, shown in \Cref{algorithm_q_sampling}, for the robust convex optimization problem in \Cref{feasibility_problem}. The main idea behind the algorithm is to build a stochastic subgradient oracle $\mathcal{O}_g$ by performing $\ell_1$-sampling from the $md$-dimensional vector $(\nabla_u f_1(x,u_1),\dots,\nabla_u f_m(x,u_m))$ and construct a stochastic subgradient using the sampled entries. For such, our algorithm makes use of the following quantum state preparation subroutine by \cite{hamoudi2022preparing}. We will also use a subroutine to estimate the $\ell_1$-norm of a vector due to \cite{Hamoudi2019}. The algorithm's correctness and runtime are proven in \Cref{correctness_runtime_q_sampling}. We shall assume quantum access to the noise memory as in \Cref{noise_memory} and the oracles $\mathcal{O}_{\nabla}$, $\mathcal{Q}_{\nabla}$, $\mathcal{O}_{\mathcal{P}}$, $\mathcal{O}_\epsilon$.
\begin{fact}[{\citep[Theorem~1]{hamoudi2022preparing}}]\label{quantum_multi_sample}
    Consider $s,n\in\mathbb{N}$ with $1\leq s\leq n$, $\delta\in(0,1)$, and a non-zero vector $u\in\mathbb{R}^n_{\geq 0}$. There is a quantum algorithm that outputs $s$ copies of the state $\sum_{i=1}^n \sqrt{\frac{u_i}{\|u\|_1}}|i\rangle$ with probability at least $1-\delta$ and in $O(\sqrt{sn}\log(1/\delta))$ time.
\end{fact}

\begin{fact}[{\citep[Lemma~3.6]{Hamoudi2019}}]\label{quantum_l1_norm}
    Consider $\delta,\nu\in(0,1)$ and a non-zero vector $u\in\mathbb{R}^n$ such that $\|u\|_\infty \leq M$ for some $M\in\mathbb{R}_{>0}$. Given an evaluation oracle to $u$ in time $\mathcal{T}_u$, there is a quantum algorithm that outputs $\Gamma\in\mathbb{R}_{>0}$ such that $|\Gamma - \|u\|_1| \leq \nu\|u\|_1$ with probability $1-\delta$ and in $O(\mathcal{T}_u\nu^{-1}\sqrt{nM/\|u\|_1}\log(1/\delta))$ time.
\end{fact} 

\begin{algorithm}[h]
\caption{Hybrid quantum-classical online sampling-based dual subgradient robust optimization algorithm}
\label{algorithm_q_sampling}
\begin{algorithmic}[1]
    \Require Accuracy $\epsilon>0$, failure probability $\delta\in(0,1)$, parameters $D,G_1,G_2,G_\infty,F$;
    \State Set $T = \big\lceil\frac{1}{\epsilon^2} \max\big\{4F\log(\frac{m}{\delta}), \frac{225D^2}{16}\big(G_2^2 + \frac{G_1G_\infty - G_2^2}{s}\big) \big\} \big\rceil$;
    \State Set $\eta^{(t)} = \frac{4D}{3\sqrt{t+1}}\big(G_2^2 + \frac{G_1G_\infty - G_2^2}{s}\big)^{-1/2}$;
    \State Initialize $(u^{(0)}_1, \dots, u^{(0)}_m)\in\mathcal{U}^m$ and $x^{(0)}\in\mathcal{D}$ arbitrarily;
    \For {$t=0$ to $T-1$}
    \State \label{line:line5}\parbox[t]{\dimexpr\linewidth-\algorithmicindent}{Sample $s$ pairs $S^{(t)}=((i_1,j_1),\dots,(i_s,j_s))\in([m]\times[d])^s$ with probability at least $1-\delta/T$ by measuring $s$ copies of the quantum state $\sum_{i=1}^m\sum_{j=1}^d \sqrt{p^{(t)}(i,j)}\ket{i}\ket{j}$ (\Cref{quantum_multi_sample}), where}
    \begin{align*}
        p^{(t)}(i,j) = \frac{\big|\big(\nabla_u f_i(x^{(t)},u_i^{(t)})\big)_j\big|}{\sum_{k=1}^m \big\|\nabla_u f_k(x^{(t)},u_k^{(t)})\big\|_1};
    \end{align*}
    \State \label{line:line6}Get an estimate $\Gamma^{(t)}$ of $\sum_{k=1}^m \big\|\nabla_u f_k(x^{(t)},u_k^{(t)})\big\|_1$ with relative error $\frac{1}{4}$ (\Cref{quantum_l1_norm});
    \State \parbox[t]{\dimexpr\linewidth-\algorithmicindent}{Query the (classical) oracle $\mathcal O_\nabla$ with inputs $(i,j)\in S^{(t)}$, $x^{(t)}$, and $u_i^{(t)}$ to prepare $g_i^{(t)}$ as}
    \begin{align*}
        (g_i^{(t)})_j = \frac{|\{(k,\ell)\in S^{(t)}:k=i,\ell=j\}|}{s}\frac{ \operatorname{sign}\big[(\nabla_u f_i(x^{(t)},u_i^{(t)}))_j\big]}{(\Gamma^{(t)})^{-1}};
    \end{align*}
    \For {$i$ in $S^{(t)}$}
    \State $u_i^{\prime(t+1)}\gets u_i^{(t)} + \eta^{(t)} g_i^{(t)}$;
    \State $u_i^{(t+1)}\gets\mathcal{P}_{\mathcal{U}}(u_i^{\prime(t+1)})$; \Comment{Update noise memory}
    \EndFor
    \State $x^{(t+1)}\gets\mathcal{O}_\epsilon(u_1^{(t+1)},\dots,u_m^{(t+1)})$;
    \If {oracle declares infeasibility}
    \Return \texttt{INFEASIBLE};
    \EndIf
    \EndFor
    \Ensure $\bar{x}=\frac{1}{T}\sum_{t=1}^T x^{(t)}$;
\end{algorithmic}
\end{algorithm}
\begin{theorem}\label{correctness_runtime_q_sampling}
    For $i\in[m]$, let $f_i: \mathcal{D}\times \mathcal{U}\to \mathbb{R}$ be a convex function over variables from a convex set $\mathcal{D}\subseteq\mathbb{R}^n$ and parameters from a convex set $\mathcal{U}\subseteq\mathbb{R}^d$ with diameter $\max_{u,v\in\mathcal{U}}\| u-v\|_2\leq D$. Let $F,G_2\in\mathbb{R}_{>0}$ be constants such that $|f_i(x,u)| \leq F$ and $\|\nabla_u f_i(x,u)\|_2\leq G_2$ for all $x\in\mathcal{D}$ and $u\in\mathcal{U}$. Let also $G_1^{(i)}\in\mathbb{R}_{>0}$ be constants for $i\in[m]$ such that $\|\nabla_u f_i(x,u)\|_1\leq G_1^{(i)}$ for all $x\in\mathcal{D}$ and $u\in\mathcal{U}$, and define $G_1 := \sum_{i=1}^m G_1^{(i)}$ and $G_\infty := \max_{i\in[m]} G_1^{(i)}$. Let $s\in\mathbb{N}$ and 
    \begin{align*}
        T = \left\lceil\frac{1}{\epsilon^2} \max\left\{4F\log\Big(\frac{m}{\delta}\Big), \frac{225D^2}{16}\left(G_2^2 + \frac{G_1G_\infty - G_2^2}{s}\right) \right\} \right\rceil.
    \end{align*}
    {\rm \Cref{algorithm_q_sampling}} either returns a $3\epsilon$-approximate solution to the robust program with probability at least $1-\delta$ or correctly declares that it is infeasible. Its runtime is
    \begin{align*}
        \widetilde{O}\left( \frac{(\sqrt{smd} + s)\mathcal{T}_\nabla + \min\{m,s\}\mathcal{T}_{\mathcal{P}} + \mathcal{T}_\epsilon}{\epsilon^2} \max\left\{F, D^2\left(G_2^2 + \frac{G_1G_\infty - G_2^2}{s}\right)\right\}\right).
    \end{align*}
    In particular, if $s = \lceil G_1G_\infty/G_2^2\rceil$, then the runtime is
    \begin{align*}
        \widetilde{O}\left( \frac{(\sqrt{G_1G_\infty}/G_2)\sqrt{md}\mathcal{T}_\nabla + \min\{G_1G_\infty/G_2^2,m\}\mathcal{T}_{\mathcal{P}} + \mathcal{T}_\epsilon}{\epsilon^2} \max\{F, D^2G_2^2\}\right).
    \end{align*}
\end{theorem}
\begin{proof}
    The correctness of the algorithm comes from \Cref{thr:correctness_general} by proving that the sampled stochastic subgradients satisfy the requirements of the oracle $\mathcal{O}_g$ from \Cref{eq:gradient_oracle}. Indeed, note that $\mathbb{E}_t[| \{(k,\ell)\in S^{(t)}:k=i,\ell=j\}|] = sp^{(t)}(i,j)$ according to the multinomial distribution and also that $\Gamma^{(t)} = \lambda^{(t)}\sum_{k=1}^m \|\nabla_u f_k(x^{(t)}, u_k^{(t)})\|_1$, where $\lambda^{(t)}\in\mathbb{R}$ is such that $|\lambda^{(t)} - 1| \leq 1/4$. Thus the random variables $g_i^{(t)}\in\mathbb{R}^d$ from \Cref{algorithm_q_sampling} satisfy
    \begin{align*}
        \mathbb{E}_t[g_i^{(t)}] = \lambda^{(t)}\nabla_u f_i(x^{(t)},u_i^{(t)}).
    \end{align*}
    Moreover, using that 
    \begin{align*}
        \mathbb{E}_t[|\{(k,\ell)\in S^{(t)}:k=i,\ell=j\}|^2] = sp^{(t)}(i,j)(1-p^{(t)}(i,j)) + s^2p^{(t)}(i,j)^2
    \end{align*}
    by the multinomial distribution, and that
    \begin{align*}
        \Gamma^{(t)} p^{(t)}(i,j) = \lambda^{(t)}\big|(\nabla_u f_i(x^{(t)},u_i^{(t)}))_j\big|,
    \end{align*}
    we obtain
    \begin{align*}
        \frac{\mathbb{E}_t[\|g_i^{(t)}\|_2^2]}{(\lambda^{(t)})^2} &= \frac{1}{(\lambda^{(t)})^2}\sum_{j=1}^d \frac{\mathbb{E}_t[|\{(k,\ell)\in S^{(t)}:k=i,\ell=j\}|^2]}{s^2(\Gamma^{(t)})^{-2}}\\
        &= \frac{1}{(\lambda^{(t)})^2}\sum_{j=1}^d (\Gamma^{(t)})^2p^{(t)}(i,j)^2 + \frac{1}{s}(\Gamma^{(t)})^2 p^{(t)}(i,j)(1-p^{(t)}(i,j))\\
        &= \big\|\nabla_u f_i(x^{(t)}, u_i^{(t)})\big\|_2^2 + \frac{1}{s}\sum_{j=1}^d \frac{1-p^{(t)}(i,j)}{p^{(t)}(i,j)}\big|(\nabla_u f_i(x^{(t)},u_i^{(t)}))_j\big|^2\\
        &= \left(\!1 - \frac{1}{s}\!\right)\!\big\|\nabla_u f_i(x^{(t)}, u_i^{(t)})\big\|_2^2 + \frac{\big\|\nabla_u f_i(x^{(t)}, u_i^{(t)})\big\|_1}{s}\sum_{k=1}^m \!\big\|\nabla_u f_k(x^{(t)}, u_k^{(t)})\big\|_1\\
        &\leq G_2^2 + \frac{G_1G_\infty - G_2^2}{s}.
    \end{align*}
    Therefore we just set $\widetilde{G}_2 := (1+\nu)\big(G_2^2 + \frac{G_1G_\infty - G_2^2}{s}\big)$ and $\nu=1/4$ in \Cref{thr:correctness_general}.
    
    We now consider the runtime of the algorithm:
    \begin{itemize}
        \item Sampling $Ts$ numbers from $\sum_{i=1}^m\sum_{j=1}^d \sqrt{p^{(t)}(i,j)}|i,j\rangle$ takes $\widetilde{O}(T\sqrt{smd}\mathcal{T}_\nabla)$ time;
        \item Estimating $\sum_{k=1}^m \big\|\nabla_u f_k(x^{(t)},u_k^{(t)})\big\|_1$ takes $\widetilde{O}(T\sqrt{md}\mathcal{T}_\nabla)$ time;
        \item Defining the vectors $g^{(t)}_i$ and obtaining $u_i^{(t)} + \eta^{(t)}g_i^{(t)}$ takes $O(Ts\mathcal{T}_\nabla)$ time;
        \item Using the projector oracles $\mathcal{P}_{\mathcal{U}}$ takes $O(T\min\{m,s\}\mathcal{T}_{\mathcal{P}})$ time;
        \item Using the optimization oracles $\mathcal{O}_\epsilon$ takes $O(T\mathcal{T}_\epsilon)$ time.
    \end{itemize}
    All in all, plugging the value $T = O\Big(\frac{1}{\epsilon^2} \max\big\{F\log(\frac{m}{\delta}), D^2\big(G_2^2 + \frac{G_1G_\infty - G_2^2}{s}\big) \big\} \Big)$, we get
    \begin{align*}
        &\widetilde{O}\Big(T\big((\sqrt{smd} + s)\mathcal{T}_\nabla + \min\{m,s\}\mathcal{T}_{\mathcal{P}} + \mathcal{T}_\epsilon\big)\Big) \\
        &= \widetilde{O}\left( \frac{(\sqrt{smd} + s)\mathcal{T}_\nabla + \min\{m,s\}\mathcal{T}_{\mathcal{P}} + \mathcal{T}_\epsilon}{\epsilon^2} \max\left\{F, D^2\left(G_2^2 + \frac{G_1G_\infty - G_2^2}{s}\right)\right\}\right).  \qedhere
    \end{align*}
\end{proof}

\begin{remark}
    It is possible to consider an equivalent classical-sampling algorithm to {\rm \Cref{algorithm_q_sampling}} which samples from the vector $(\nabla_u f_1(x,u_1),\dots,\nabla_u f_m(x,u_m))$ classically. For such, a sampling data structure must be constructed for every time step $t\in[T]$, which can be done in time $O(md\mathcal{T}_\nabla)$~{\rm \citep{Vose1991}}. This is, unfortunately, not asymptotically better than computing all the subgradients exactly. An advantage can be obtained, though, in the complexity for oracle $\mathcal{O}_{\mathcal{P}}$ similarly to the hybrid quantum-classical algorithm.
\end{remark}

\begin{remark}
    All steps in {\rm \Cref{algorithm_q_sampling}} are entirely classical apart from {\rm \Cref{line:line5,line:line6}}. 
\end{remark}

\section{Applications and examples}

In this section, we provide a few examples to which our hybrid quantum-classical algorithm can be applied: robust linear and semidefinite programming problems. We shall focus on ellipsoidal uncertainty, being the most common model of data uncertainty since it has a simple parametric representation, can be easily handled numerically, and there are probabilistic arguments which allow the replacement of stochastic uncertainty with ellipsoidal deterministic uncertainty (see~\citep{ben1998robust,ben1999robust} for more information). Moreover, this choice also allows us to compare our results with that of \cite{Ben-Tal2015a}.

\subsection{Robust linear programming}

A linear program (LP) in standard form is
\begin{align}
    \label{reduced_LP}
    \begin{aligned}
    & \exists ?
    & & x\in\mathbb R^n\\
    & \text{s.t.}
    & & a_i^\top x -  b_i\leq 0, \quad i\in[m],\\
    \end{aligned}
\end{align}
where $a_1,\dots,a_m\in\mathbb R^n$ and $b\in\mathbb R^m$.  We assume below that $\mathcal{D}\subseteq\{x\in\mathbb{R}^n:\|x\|_1\leq 1\}$. The robust counterpart of \Cref{reduced_LP} with ellipsoidal noise takes the form of
\begin{align}
    \label{robust_LP}
    \begin{aligned}
    & \exists ?
    & & x\in\mathbb R^n\\
    & \text{s.t.}
    & & \left(a_i + P_i u_i\right)^\top x -  b_i\leq 0, \hspace{1cm} \forall u_i\in\mathcal U, \quad i\in[m],\\
    \end{aligned}
\end{align}
where $P_i\in\mathbb R^{n\times d}$ controls the shape of the ellipsoidal uncertainty and $\mathcal U=\{u\in\mathbb R^d: \| u\|_2\leq 1\}$ is the $d$-dimensional unit ball. The robust optimization problem in \Cref{robust_LP} is a second-order conic program that can be efficiently solved~\citep{ben2009robust,bertsimas2011theory}. An advantage of solving a robust LP using its original formulation stems from the fact that many cases of interest present some special combinatorial structure which allows for highly efficient solvers compared to generic LP solvers, e.g., network flow problems. Since such special structure is lost for their robust counterparts, solving the robust problem using an oracle to the original formulation might be preferable in these cases.

\Cref{algorithm_q_sampling} can be used to solve the robust LP in \Cref{robust_LP} since the constraints are linear with respect to the noise vectors $u_i$. We have that $\nabla_u f_i(x,u) = P_i^\top x$, and so \Cref{algorithm_q_sampling} updates the variables $u_i^{(t)}$ according to
\begin{align*}
    u_i^{(t+1)} \gets \frac{u_i^{(t)} + \eta^{(t)} g_i^{(t)}}{\max\big\{1,\|u_i^{(t)} + \eta^{(t)} g_i^{(t)}\|_2\big\}}, \quad \text{where}\quad \mathbb{E}[g_i^{(t)}] = \lambda^{(t)} P_i^\top x^{(t)}.
\end{align*}
Applying \Cref{correctness_runtime_q_sampling} to such robust LPs, we obtain the following corollary.
\begin{corollary}\label{robust_LP_runtime}
    {\rm \Cref{algorithm_q_sampling}} gives an $3\epsilon$-solution to {\rm \Cref{robust_LP}} in time
    \begin{align*}
        \widetilde O\left(\left(\frac{\sqrt{G_1G_\infty}}{G_2}\sqrt{md}\mathcal T_\nabla + \min\bigg\{m,\frac{G_1G_\infty}{G_2^2}\bigg\}\mathcal T_{\mathcal P} + \mathcal T_\epsilon\right)\frac{G_2^2}{\epsilon^2}\right),
    \end{align*}
    where $G_2 = \max_{i\in[m]}\|P_i\|_2$, $G_\infty = \max_{i\in[m]}\|P_i\|_\infty$, and $G_1 = \sum_{i=1}^m \|P_i\|_\infty$. In particular, if $P_i = P$ for all $i\in[m]$, then the runtime is
    \begin{align*}
        \widetilde O\left(\left(\frac{\|P\|_\infty}{\|P\|_2}
        m\sqrt{d}\mathcal T_\nabla + \min\bigg\{1,\frac{\|P\|_\infty^2}{\|P\|_2^2}\bigg\}m\mathcal T_{\mathcal P} + \mathcal T_\epsilon\right)\frac{\|P\|_2^2}{\epsilon^2}\right).
    \end{align*}
\end{corollary}
\begin{proof}
Note that for all $u_i\in\mathcal U$ and $\| x\|_2 \leq \|x\|_1 \leq 1$,
\begin{align*}
    \| \nabla_u f_i(x, u_i)\|_2 &= \|P_i^\top x\|_2 \leq \| P_i\|_2 \| x\|_2 \leq \|P_i\|_2,\\
    \| \nabla_u f_i(x, u_i)\|_1 &= \|P_i^\top x\|_1 \leq \| P_i^\top\|_1 \|x\|_1 \leq \|P_i\|_\infty.
\end{align*}
Therefore $G_\infty = \max_{i\in[m]}\|P_i\|_\infty$, $G_1 = \sum_{i=1}^m\|P_i\|_\infty$, and $G_2 = \max_{i\in[m]}\|P\|_2$. Using this in \Cref{correctness_runtime_q_sampling} and also that $D = 2$ and that
\begin{align*}
    F = \operatorname*{\max}_{\|x\|_1,\|u\|_2\leq 1}|x^\top P_i u| \leq \operatorname*{\max}_{\|x\|_1,\|u\|_2\leq 1}\|x\|_2 \|P_i\|_2 \|u\|_2 \leq G_2
\end{align*}
(see \Cref{rem:remark1}), we obtain the desired runtime. 
\end{proof}

\begin{remark}\label{remark:LP}
    The corresponding classical runtime of {\rm \citep[Corollary~1]{Ben-Tal2015a}} for solving the robust LP from {\rm \Cref{robust_LP}} with $P_i=P$, $i\in[m]$, is
    \begin{align*}
        \widetilde O\left(\left(md\mathcal T_{\nabla} + m\mathcal T_{\mathcal P}+ \mathcal T_\epsilon\right)\frac{\|P\|^2_2}{\epsilon^2}\right).
    \end{align*}
    Comparing between the classical and quantum runtime for solving robust LPs, we see an improvement in the dimension $d$ if $\|P\|_\infty/\|P\|_2 = o(\sqrt{d})$.
\end{remark}

One practical example of robust LP in finance is the global maximum return portfolio (GMRP) robust optimization problem~\citep{goldfarb2003robust}. The problem considers $n$ assets in a setting with $m$ markets, each having a different expected return vector. Let $x\in\mathbb{R}^n_{\geq 0}$ such that $\sum_{i=1}^n x_i = 1$ be the portfolio vector. In each market $i\in[m]$, the expected return of the assets is given by the vector $r_i\in\mathbb{R}^n$. The GMRP optimization problem aims to find a portfolio that maximizes the return (without considering the variance), and can be represented in the following linear program, where $\mathcal{D} \subseteq \{x\in\mathbb{R}^n_{\geq 0}:\|x\|_1 = 1\}$ and $c>0$ is a minimum expected return, 
\begin{eqnarray}\nonumber
\exists ?
& & x\in\mathcal{D}\\
\nonumber\text{s.t.} & & r_i^\top x\geq c, \quad \forall i\in[m].
\end{eqnarray}
However, the returns $r_1,\dots,r_m$ are unknown in practice and need to be estimated over some period of time, thus being extremely noisy~\citep{chopra2013effect}. Therefore, instead of assuming that $r_1,\dots,r_m$ are known perfectly, we assume that they belong to some convex uncertainty set. More specifically, the exact return $r_i$ deviates from a given estimation $\widetilde{r}_i$ due to an ellipsoidal noise such that $r_i = \widetilde{r}_i + \kappa_i S_i^{1/2}u_i$, where, for $i\in[m]$, $\kappa_i\in\mathbb{R}_{> 0}$ is a constant, $S_i\in\mathbb{R}^{n\times d}$ controls the shape of the ellipsoid, and $u_i\in\mathcal{U} = \{u\in\mathbb{R}^d:\|u\|_2\leq 1\}$ is the usual $d$-dimensional unit ball. The robust framework allows to include the return covariance matrix $\Sigma_i$ for the $i$-th market (analogous to the Markovitz framework~\citep{HarryMarkowitz1952}) by setting $S_i = \Sigma_i$~\citep{lobo2000worst}. The problem can now be rewritten as 
\begin{eqnarray}\nonumber
\exists ?
& & x\in\mathcal{D}\\
\nonumber \text{s.t.} & & \big(\widetilde{r}_i + \kappa_i S_i^{1/2}u_i\big)^\top x \geq c, \hspace{1cm} \forall u_i\in\mathcal{U},\quad i\in[m].
\end{eqnarray}
In this case, $f_i(x,u_i) = c - \big(\widetilde{r}_i + \kappa_i S_i^{1/2}u_i\big)^\top x$ and so $\nabla_u f_i(x,u_i) = - \kappa_i (S^{1/2}_i)^\top x$. \Cref{robust_LP_runtime} readily gives the runtime of \Cref{algorithm_q_sampling} applied to the above robust GMRP problem just by replacing $P_i$ with $\kappa_i S_i^{1/2}$.

\subsection{Robust Semidefinite Programming}

Let $\mathbb{S}_+^n := \{X\in\mathbb{R}^{n\times n}: X\succeq 0\}$ be the cone of $n\times n$ positive semidefinite matrices and $\mathcal{D} \subseteq \{X\in\mathbb{S}_+^n: \|X\|_F \leq 1\}$. A semidefinite program (SDP) in standard form is given by 
\begin{eqnarray}
\nonumber\exists? & & X\in\mathcal{D}\\
\nonumber\text{s.t.} & & A_i\bullet X\leq b_i, \quad i\in[m],
\end{eqnarray}
where $A_i\in\mathbb R^{n\times n}$ and $b_i\in\mathbb{R}$, for $i\in[m]$, are given, and $A\bullet X := \operatorname{Tr}[A^\top X]$ denotes the inner product between matrices $A$ and $X$. A robust counterpart of the above SDP with an ellipsoidal uncertainties is of the form
\begin{align}
    \label{robust_SDP}
    \begin{aligned}
    & \exists ?
    & & X\in\mathcal{D} \\
    & \text{s.t.}
    & & \left(A_i + \sum_{j=1}^d u_{ij} P_j\right)\bullet X - b_i\leq 0, \quad \forall u_i\in\mathcal U, ~i\in[m],
    \end{aligned}
\end{align}
where $P_j\in\mathbb R^{n\times n}$ are fixed for $j\in[d]$ and $\mathcal{U}= \{u\in\mathbb R^d : \| u\|_2\leq 1\}$. It is known that the robust counterpart of an SDP is NP-hard even with ellipsoidal uncertainties~\citep{ben2002robust,ben2009robust}. It is possible to apply \Cref{algorithm_q_sampling} to the robust SDP since the constraints are linear with respect to the noise vectors. Then $\nabla_u f_i(X,u) = (P_1\bullet X,\dots, P_d\bullet X)$ and so \Cref{algorithm_q_sampling} updates the variables $u_i^{(t)}$ according to
\begin{align*}
    u^{(t+1)}_i \gets \frac{u_i^{(t)} + \eta^{(t)}g_i^{(t)}}{\max\big\{1,\|u^{(t)}_i + \eta^{(t)}g_i^{(t)}\|_2\big\}}, \quad \text{where}\quad \mathbb{E}[g_i^{(t)}] = (P_1\bullet X^{(t)},\dots, P_d\bullet X^{(t)}).
\end{align*}
Applying \Cref{correctness_runtime_q_sampling} to the above robust SDPs, we obtain the following corollary.
\begin{corollary}\label{robust_SDP_runtime}
    {\rm \Cref{algorithm_q_sampling}} gives an $3\epsilon$-solution to  {\rm \Cref{robust_SDP}} in time 
    \begin{align*}
        \widetilde O\left(\left(\frac{\Sigma}{\sigma} m\sqrt{d}\mathcal T_\nabla + m\mathcal T_{\mathcal P}+ \mathcal T_\epsilon\right)\frac{\sigma^2}{\epsilon^2}\right),
    \end{align*}
    where $\sigma^2 = \sum_{j=1}^d \| P_j\|_F^2$ and $\Sigma = \sum_{j=1}^d \| P_j\|_F$.
\end{corollary}
\begin{proof}
     Note that for all $u_i\in\mathcal U$ and $\|X\|_F\leq 1$, 
     \begin{align*}
        \| \nabla f_i(X, u_i)\|_2^2 &= \sum_{j=1}^d \vert P_j\bullet X\vert^2\leq \sum_{j=1}^d \| P_j\|^2_F \| X\|^2_F\leq \sum_{j=1}^d \| P_j\|^2_F = \sigma^2, \\
        \| \nabla f_i(X, u_i)\|_1 &= \sum_{j=1}^d \vert P_j\bullet X\vert\ \leq \sum_{j=1}^d \| P_j\|_F \| X\|_F\leq \sum_{j=1}^d \| P_j\|_F = \Sigma.
    \end{align*}
    Therefore $G_1 = m\Sigma$, $G_\infty = \Sigma$, and so $G_1G_\infty/G_2^2 = m\Sigma^2/\sigma^2$. Using this in \Cref{correctness_runtime_q_sampling} and also that $D = 2$ and that
    \begin{align*}
        F = \operatorname*{\max}_{\|u\|_2,\|X\|_F \leq 1} \left|\sum_{j=1}^d u_j P_j\bullet X \right| \leq \operatorname*{\max}_{\|u\|_2,\|X\|_F \leq 1} \|u\|_2\sqrt{\sum_{j=1}^d |P_j\bullet X|^2} \leq \sigma
    \end{align*}
    (see \Cref{rem:remark1}), we get the desired runtime. 
\end{proof}

\begin{remark}
    The corresponding classical runtime of {\rm \citep[Corollary~4]{Ben-Tal2015a}} for solving the robust SDP from {\rm \Cref{robust_SDP}} is
    \begin{align*}
        \widetilde O\left(\left(md\mathcal T_{\nabla} + m\mathcal T_{\mathcal P}+ \mathcal T_\epsilon\right)\frac{\sigma^2}{\epsilon^2}\right).
    \end{align*}
    Similar to the analysis in {\rm \Cref{remark:LP}}, a speedup on the dimension $d$ is possible if $\Sigma/\sigma = o(\sqrt{d})$, e.g., if $\|P_k\|_F \gg \max_{j\in[d]\setminus\{k\}}\|P_j\|_F$ for some $k\in[d]$. 
\end{remark}

\cite{ben1997robust,ben2000robust} considered the problem of (static) structural design as an example for a robust SDP, where the aim is to find a mechanical construction able to withstand a given external load. In particular, they were interested in the Truss Topology Design (TTD) problem. A truss is a structure that consists of thin elastic bars linked with each other at nodes and deforms at the weight of a given load, until tension that causes the deformation cancels out with the load. The compliance --- potential energy stored in the deformed truss --- measures the stiffness of the truss. The lesser the compliance, the more stiff the truss is, and hence a greater ability to withstand the weight of the load. 

In the TTD problem, we are given the geometry of the nodal set $v_1,\dots, v_m\in\mathbb R^n$ and a worst-case external load $f\in \mathbb{R}^n$. The optimization variables are the compliance $\tau$ and volume $t_i$ of the $i$-th bar, $i\in[m]$. The vector $t = (t_1, \dots, t_m) \in\mathbb R^m$ is known to belong to a polytope $\mathcal E$ of design restrictions, e.g., the simplex $\{t\in\mathbb{R}_{\geq 0}^m: \sum_{i=1}^m t_i = V\}$ for some volume $V\in\mathbb{R}_{>0}$. The TTD problem can be expressed as the following SDP (see~\citep[Eq.~(35)]{ben2000robust}):
\begin{align}
    \label{TTD}
    \begin{aligned}
    \operatorname*{minimize}_{t\in\mathbb{R}^m,\tau\in\mathbb{R}}
    & ~~ \tau\\
    \text{subject to} & ~~ \begin{bmatrix}\tau & f^\top\\ f & \sum_{i=1}^m t_i v_i v_i^\top\end{bmatrix}\succeq 0,\\
    & ~~\sum_{i=1}^m t_i = V,\\
    & ~~ t_i \geq 0, \quad i\in[m].
    \end{aligned}
\end{align}
Notice that \Cref{TTD} can be expressed in the form of a dual SDP as
\begin{eqnarray}
\begin{aligned}\nonumber
\operatorname*{maximize}_{y\in\mathbb{R}^{m+1}}
& ~~ b^\top y \\
\nonumber\text{subject to} & ~~ \sum_{i=1}^{m+1} y_i A_i - C \preceq 0,\\
& ~~ \sum_{i=1}^m y_i = -V,\\
& ~~ y_i \leq 0, \quad i\in[m],
\end{aligned}
\end{eqnarray}
where 
$$b = \begin{bmatrix}0\\ \vdots\\ 0\\ 1\end{bmatrix}, ~~ y = \begin{bmatrix} -t_1\\ \vdots \\ -t_m \\ -\tau\end{bmatrix}, ~~ C = \begin{bmatrix} 0 & f^\top\\ f & 0\end{bmatrix}, ~~ A_i = \begin{bmatrix} 0 & 0 \\ 0 & v_i v_i^\top\end{bmatrix} ~\text{for}~ i\in[m], ~~ A_{m+1} = \begin{bmatrix} 1 & 0 \\ 0 & 0_n \end{bmatrix}.$$
The corresponding primal SDP is then 
\begin{eqnarray}
\begin{aligned}\nonumber
\operatorname*{minimize}_{z\in\mathbb{R},X\in\mathbb{S}_+^{n+1}}
& ~~ C\bullet X + Vz \\\nonumber
\text{subject to} & ~~ A_i\bullet X - z\leq 0, \quad \forall i\in[m],\\\nonumber
& ~~ A_{m+1}\bullet X = 1,
\end{aligned}
\end{eqnarray}
or, in its feasibility form similar to \Cref{feasibility_problem} and already taking the feasibility domain $\mathcal{D} = \{X\in\mathbb{S}_+^{n+1}:\|X\|_F \leq 1, X_{11} = 1\}\times[-1,1]$ (where the constraint $X_{11} = 1$ comes from $A_{m+1}\bullet X = 1$),
\begin{align}
    \begin{aligned}\label{TTD_SDP}
    \exists ?
    & ~~ (X,z)\in\mathcal{D} \\
    \text{s.t.} & ~~  \lambda - C\bullet X - Vz\leq 0,\\
    & ~~ A_i \bullet X - z\leq 0, \quad \forall i\in[m].
    \end{aligned}
\end{align}
The ``standard'' TTD problem only takes into consideration a finite set of relevant external loads, which is usually the engineer's guess of relevant scenarios. However, even a small load not belonging to the relevant set might lead to a large deformation of the truss. It is then important to consider a robust TTD problem with a larger set of external loads containing the original one. For simplicity, we shall take such larger set as the ellipsoid centered at the origin, $Q\mathcal{U} := \{Qu\in\mathbb{R}^n:u\in\mathbb{R}^d,\|u\|_2\leq 1\}$, where $Q\in\mathbb{R}^{n\times d}$ is a scale matrix and $\mathcal{U}$ is the usual $d$-dimensional unit ball. The robust version of \Cref{TTD_SDP} is thus (see~\citep[Section~4]{ben1997robust}):
\begin{eqnarray}\begin{aligned}\nonumber
\exists ?
& ~~ (X,z)\in \mathcal{D}\\
\text{s.t.} & ~~ \lambda - \left(\sum_{j=1}^d u_j P_j\right)\bullet X - Vz \leq 0,\quad \text{for }u\in\mathcal{U},\\
& ~~ A_i \bullet X - z\leq 0, \quad \forall i\in[m],
\end{aligned}
\end{eqnarray}
where $P_j = \bigl(\begin{smallmatrix}
0&Q_j^\top \\ Q_j&0
\end{smallmatrix} \bigr)\in\mathbb R^{(n+1)\times (n+1)}$ and $Q_j$ denotes the $j$-th column of $Q$, since 
\begin{align*}
    \begin{bmatrix}
        0 & (Qu)^\top \\
        Qu & 0
    \end{bmatrix} = 
    \begin{bmatrix}
        0 & \sum_{j=1}^d Q_{1j}u_j & \sum_{i=j}^d Q_{2j}u_j & \cdots & \sum_{j=1}^d Q_{nj}u_j \\
        \sum_{j=1}^d Q_{1j}u_j & 0 & 0 & \cdots & 0\\
        \sum_{j=1}^d Q_{2j}u_j & 0 & 0 & \cdots & 0\\
        \vdots & \vdots & \vdots & \ddots & \vdots\\
        \sum_{j=1}^d Q_{nj}u_j & 0 & 0 & \cdots & 0
    \end{bmatrix} .
\end{align*}
In this case, $f(X,z, u)=\lambda - \big(\sum_{j=1}^d u_j P_j\big)\bullet X - Vz$ and so $\nabla_u f(X,z, u) = -(P_1\bullet X,\dots,P_d\bullet X)$. The runtime of \Cref{algorithm_q_sampling} applied to the above robust TTD problem follows readily from \Cref{robust_SDP_runtime} with $m=1$, since the ellipsoidal noise affects only one constraint, and
\begin{align*}
    \sigma^2 = \sum_{j=1}^d\|P_j\|_F^2 = 2\sum_{j=1}^d \|Q_j\|^2_2 = 2\|Q\|_F^2,\quad \quad \quad
    \Sigma = \sum_{j=1}^d\|P_j\|_F = \sqrt{2}\sum_{j=1}^d \|Q_j\|_2.
\end{align*}
\section{Discussion and conclusion}
We studied the problem of robust online convex optimization, where the noise vectors belong to a convex uncertainty set, and all constraints have to be satisfied for all the noise vectors. The online essence of this problem arises from the existence of an optimization oracle that returns a solution based on the noise vectors that are being updated at every iteration. 

We devised a hybrid quantum-classical meta-algorithm for the aforementioned problem that runs in time 
$$\widetilde{O}\left( \frac{(\sqrt{G_1G_\infty}/G_2)\sqrt{md}\mathcal{T}_\nabla + \min\{G_1G_\infty/G_2^2,m\}\mathcal{T}_{\mathcal{P}} + \mathcal{T}_\epsilon}{\epsilon^2} \max\{F, D^2G_2^2\}\right)$$ 
as compared to the corresponding classical meta-algorithm by \cite{Ben-Tal2015a} that runs in time $O\big(\frac{D^2G_2^2}{\epsilon^2}( md\mathcal{T}_\nabla + m\mathcal{T}_{\mathcal{P}} + \mathcal{T}_\epsilon)\big)$. If $G_1 G_\infty$ is comparable to $G_2^2$, then we obtain a quadratic improvement in terms of the number of noise vectors $m$ and their dimension $d$. If, however, $G_1 G_\infty = O(mG_2^2)$, then we only achieve a quadratic speedup in $d$. The speedup is due to techniques such as quantum state preparation, quantum norm estimation, and quantum multi-sampling. Even if, on the one hand, the quantum advantages provided are at most quadratic and depend on quantities like $G_1$ and $G_\infty$, on the other hand, we argue that our results also show that the meta-framework by Ben-Tal \emph{et al.} cannot be easily quantized to obtain large speedups, which is a useful consideration, since this framework could be see as a relatively
natural candidate for quantization.

The time $\mathcal T_{\mathcal{P}}$ required by the projection oracle to map the vectors $u^{\prime(t)}_1, \dots, u^{\prime(t)}_m$ back to the convex set $\mathcal U$ depends on the type/shape of $\mathcal U$. If $\mathcal{U} = \{u\in\mathbb R^d : \| u\|_2\leq 1\}$ is the Euclidean ball, then the projection of a vector $x\in\mathbb R^d$ onto $\mathcal U$
is linear in the number of nonzero entries of $x$ assuming without loss of generality that the Euclidean ball is centered at the origin. Projection of a vector $x$ onto the $d$-dimensional unit simplex or $\ell_1$-ball requires $O(d\log d)$ time~\citep{duchi2008efficient}. Lastly, if $\mathcal U$ is a polytope, the problem of projection onto a polytope $\{x\in\mathbb R^d: Ax\leq b\}$ for some $A\in\mathbb R^{m\times d}$ and $b\in\mathbb R^m$ can be represented as a conic quadratic programming problem, which can be solved using the interior point method whose time complexity is dominated by $\sqrt{md}(d^2 + m)$~\citep{ben2001lectures}. In cases where the number of constraints is beyond $\poly(m)$, the interior point method is no longer useful. Other approaches such as the ellipsoid method may be considered but the complexity scales polynomially with $d$. 

For the error analysis, we reproved, in \Cref{lem:regret_stochastic_gradient}, an upper bound on the regret of the online stochastic (sub)gradient descent~\citep{Hazan2016b}, where the assumption on the subgradient oracle is slightly changed. Using \Cref{lem:regret_stochastic_gradient}, we showed that our {hybrid quantum-classical algorithm produces a $3\epsilon$-approximate solution to the robust optimization problem. We then applied our algorithm to solving robust LPs and robust SDPs whose noise vectors belong to an ellipsoidal uncertainty set, including explicit examples in finance and structural design. 

For future work, one can consider more general robust optimization problems where the uncertainty set is not necessarily convex, e.g., robust quadratic programming. Some applications of quadratic programs in finance include the Global Minimum Variance Problem (GMVP)~\citep{goldfarb2003robust} and Markowitz mean-variance portfolio optimization problem~\citep{HarryMarkowitz1952}. A classical meta-algorithm for such robust optimization problems has been proposed by \cite{Ben-Tal2015a}, called dual-perturbation robust optimization algorithm. It would be interesting to see if quantum techniques could improve the algorithm's total runtime. Another open question is proving a lower bound on the number of calls to the oracle $\mathcal{O}_\epsilon$.

\paragraph{Acknowledgments} 
We thank Srijita Kundu for bringing the work~\cite{Ben-Tal2015a} to our attention. 
We acknowledge valuable discussions with Naixu Guo and Yassine Hamoudi, and also thank Yassine Hamoudi for pointing out \cite{hamoudi2022preparing}.
DL acknowledges funding from the QuantERA Project QOPT. 
This research is supported by the National Research Foundation, Singapore and A*STAR under its CQT Bridging Grant and its Quantum Engineering Programme under grant NRF2021-QEP2-02-P05. JFD acknowledges funding from ERC grant No.\ 810115-DYNASNET.

\bibliography{QORO}

@article{allcock2023constant,
  doi = {10.22331/q-2024-11-20-1530},
  url = {https://doi.org/10.22331/q-2024-11-20-1530},
  title = {Constant-depth circuits for {B}oolean functions and quantum memory devices using multi-qubit gates},
  author = {Allcock, Jonathan and Bao, Jinge and Doriguello, Joao F. and Luongo, Alessandro and Santha, Miklos},
  journal = {{Quantum}},
  issn = {2521-327X},
  publisher = {{Verein zur F{\"{o}}rderung des Open Access Publizierens in den Quantenwissenschaften}},
  volume = {8},
  pages = {1530},
  month = nov,
  year = {2024}
}

@article{chen2020semialgebraic,
  title={Semialgebraic optimization for lipschitz constants of relu networks},
  author={Chen, Tong and Lasserre, Jean B and Magron, Victor and Pauwels, Edouard},
  journal={Advances in Neural Information Processing Systems},
  volume={33},
  pages={19189--19200},
  year={2020}
}

@article{geoffrey2020robust,
  title={Robust high dimensional learning for Lipschitz and convex losses},
  author={Geoffrey, Chinot and Guillaume, Lecu{\'e} and Matthieu, Lerasle},
  journal={Journal of Machine Learning Research},
  volume={21},
  number={233},
  pages={1--47},
  year={2020}
}

@inproceedings{bhowmick2021lipbab,
  title={Lipbab: Computing exact lipschitz constant of relu networks},
  author={Bhowmick, Aritra and D’Souza, Meenakshi and Raghavan, G Srinivasa},
  booktitle={International Conference on Artificial Neural Networks},
  pages={151--162},
  year={2021},
  organization={Springer}
}

@article{avant2023analytical,
  title={Analytical bounds on the local Lipschitz constants of ReLU networks},
  author={Avant, Trevor and Morgansen, Kristi A},
  journal={IEEE Transactions on Neural Networks and Learning Systems},
  volume={35},
  number={10},
  pages={13902--13913},
  year={2023},
  publisher={IEEE}
}

@article{meng2024lipschitz,
  title={Lipschitz continuity of solution multifunctions of extended $\ell_1$ regularization problems},
  author={Meng, Kaiwen and Wu, Pengcheng and Yang, Xiaoqi},
  journal={arXiv preprint arXiv:2406.16053},
  year={2024}
}

@article{van2019quantum,
  title={Quantum algorithms for zero-sum games},
  author={{{\DE{Ag}{van}} Apeldoorn}, Joran and Gily{\'e}n, Andr{\'a}s},
  journal={arXiv preprint arXiv:1904.03180},
  year={2019},
  doi={10.48550/arXiv.1904.03180}
}

@article{VanApeldoorn2020,
    title = {{Quantum SDP-Solvers: Better upper and lower bounds}},
    year = {2020},
    journal = {Quantum},
    author = {{{\DE{Ag}{van}} Apeldoorn}, Joran and Gily{\'{e}}n, András and Gribling, Sander and {de Wolf}, Ronald},
    pages = {1--69},
    volume = {4},
    doi = {10.22331/q-2020-02-14-230},
    issn = {2521327X}
}

@inproceedings{arora2009learning,
 author = {Arora, Raman},
 booktitle = {Advances in Neural Information Processing Systems},
 editor = {Y. Bengio and D. Schuurmans and J. Lafferty and C. Williams and A. Culotta},
 publisher = {Curran Associates, Inc.},
 title = {On Learning Rotations},
 volume = {22},
 year = {2009}, 
 address = {Vancouver, BC, Canada}
}

@inproceedings{Arora2016,
author = {Arora, Sanjeev and Ge, Rong and Kannan, Ravindran and Moitra, Ankur},
title = {Computing a nonnegative matrix factorization -- provably},
year = {2012},
isbn = {9781450312455},
publisher = {Association for Computing Machinery},
address = {New York, NY, USA},
doi = {10.1145/2213977.2213994},
booktitle = {Proceedings of the Forty-Fourth Annual ACM Symposium on Theory of Computing},
pages = {145–162},
numpages = {18},
location = {New York, New York, USA},
series = {STOC '12}
}

@article{arora2012multiplicative,
 author = {Arora, Sanjeev and Hazan, Elad and Kale, Satyen},
 title = {The Multiplicative Weights Update Method: a Meta-Algorithm and Applications},
 year = {2012},
 pages = {121--164},
 doi = {10.4086/toc.2012.v008a006},
 publisher = {Theory of Computing},
 journal = {Theory of Computing},
 volume = {8},
 number = {6}
}

@inproceedings{Arora2016a,
author = {Arora, Sanjeev and Kale, Satyen},
title = {A combinatorial, primal-dual approach to semidefinite programs},
year = {2007},
isbn = {9781595936318},
publisher = {Association for Computing Machinery},
address = {New York, NY, USA},
doi = {10.1145/1250790.1250823},
booktitle = {Proceedings of the Thirty-Ninth Annual ACM Symposium on Theory of Computing},
pages = {227–236},
numpages = {10},
location = {San Diego, California, USA},
series = {STOC '07}
}

@InProceedings{arunachalam2020quantum,
  title = 	 {Quantum Boosting},
  author =       {Arunachalam, Srinivasan and Maity, Reevu},
  booktitle = 	 {Proceedings of the 37th International Conference on Machine Learning},
  pages = 	 {377--387},
  year = 	 {2020},
  editor = 	 {III, Hal Daumé and Singh, Aarti},
  volume = 	 {119},
  series = 	 {Proceedings of Machine Learning Research},
  address = {Vienna, Austria}, 
  month = 	 {13--18 Jul},
  publisher =    {PMLR}
}

@INPROCEEDINGS{belmega2022online,
  author={Belmega, E. Veronica and Mertikopoulos, Panayotis and Negrel, Romain},
  booktitle={2022 20th International Symposium on Modeling and Optimization in Mobile, Ad hoc, and Wireless Networks (WiOpt)}, 
  title={Online convex optimization in wireless networks and beyond: The feedback-performance trade-off}, 
  year={2022},
  volume={},
  number={},
  pages={298-305},
  doi={10.23919/WiOpt56218.2022.9930534}
}

@article{ben1998robust,
author = {Ben-Tal, A. and Nemirovski, A.},
title = {Robust Convex Optimization},
journal = {Mathematics of Operations Research},
volume = {23},
number = {4},
pages = {769-805},
year = {1998},
doi = {10.1287/moor.23.4.769}
}

@incollection{ben2000robust,
author="Ben-Tal, Aharon
and El Ghaoui, Laurent
and Nemirovski, Arkadi",
editor="Wolkowicz, Henry
and Saigal, Romesh
and Vandenberghe, Lieven",
title="Robustness",
bookTitle="Handbook of Semidefinite Programming: Theory, Algorithms, and Applications",
year="2000",
publisher="Springer US",
address="Boston, MA",
pages="139--162",
isbn="978-1-4615-4381-7",
doi="10.1007/978-1-4615-4381-7_6"
}

@book{ben2009robust,
title = {Robust Optimization},
author = {Aharon Ben-Tal and Laurent El Ghaoui and Arkadi Nemirovski},
publisher = {Princeton University Press},
address = {Princeton},
doi = {doi:10.1515/9781400831050},
isbn = {9781400831050},
year = {2009},
lastchecked = {2024-11-14}
}

@article{Ben-Tal2015a,
    title = {{Oracle-based robust optimization via online learning}},
    year = {2015},
    journal = {Operations Research},
    author = {Ben-Tal, Aharon and Hazan, Elad and Koren, Tomer and Mannor, Shie},
    number = {3},
    pages = {628--638},
    volume = {63},
    doi = {10.1287/opre.2015.1374},
    issn = {15265463}
}

@article{ben1997robust,
author = {Ben-Tal, A. and Nemirovski, A.},
title = {Robust Truss Topology Design via Semidefinite Programming},
journal = {SIAM Journal on Optimization},
volume = {7},
number = {4},
pages = {991-1016},
year = {1997},
doi = {10.1137/S1052623495291951}
}

@article{ben1999robust,
title = {Robust solutions of uncertain linear programs},
journal = {Operations Research Letters},
volume = {25},
number = {1},
pages = {1-13},
year = {1999},
issn = {0167-6377},
doi = {https://doi.org/10.1016/S0167-6377(99)00016-4},
author = {A. Ben-Tal and A. Nemirovski}
}

@Article{ben2002robust,
author={Ben-Tal, Aharon
and Nemirovski, Arkadi},
title={Robust optimization -- methodology and applications},
journal={Mathematical Programming},
year={2002},
month={May},
day={01},
volume={92},
number={3},
pages={453-480},
issn={1436-4646},
doi={10.1007/s101070100286}
}

@book{ben2001lectures,
author = {Ben-Tal, Aharon and Nemirovski, Arkadi},
title = {Lectures on Modern Convex Optimization},
publisher = {Society for Industrial and Applied Mathematics},
year = {2001},
doi = {10.1137/1.9780898718829},
address = {USA}
}

@article{bertsimas2011theory,
author = {Bertsimas, Dimitris and Brown, David B. and Caramanis, Constantine},
title = {Theory and Applications of Robust Optimization},
journal = {SIAM Review},
volume = {53},
number = {3},
pages = {464-501},
year = {2011},
doi = {10.1137/080734510}
}

@inbook{bottou1998online, 
author = {Bottou, L\'{e}on}, title = {On-line learning and stochastic approximations}, year = {1999}, isbn = {0521652634}, publisher = {Cambridge University Press}, address = {USA}, booktitle = {On-Line Learning in Neural Networks}, pages = {9–42}, numpages = {34}}

@INPROCEEDINGS{Brandao2017,
  author={Brandao, Fernando G.S.L. and Svore, Krysta M.},
  booktitle={2017 IEEE 58th Annual Symposium on Foundations of Computer Science (FOCS)}, 
  title={Quantum Speed-Ups for Solving Semidefinite Programs}, 
  year={2017},
  volume={},
  number={},
  pages={415-426},
  doi={10.1109/FOCS.2017.45}
}

@article{brassard2011optimal,
  title={An optimal quantum algorithm to approximate the mean and its application for approximating the median of a set of points over an arbitrary distance},
  author={Brassard, Gilles and Dupuis, Frederic and Gambs, Sebastien and Tapp, Alain},
  journal={arXiv preprint arXiv:1106.4267},
  year={2011},
  doi={10.48550/arXiv.1106.4267}
}

@article{Brassard2002,
    title={Quantum amplitude amplification and estimation},
    author={Brassard, Gilles and H{\o}yer, Peter and Mosca, Michele and Tapp, Alain},
    journal={Contemporary Mathematics},
    volume={305},
    pages={53-74},
    year={2002},
    publisher={Providence, RI; American Mathematical Society; 1999},
    doi={10.1090/conm/305/05215}
}

@book{cesa2006prediction, 
author = {Cesa-Bianchi, Nicolo and Lugosi, Gabor}, title = {Prediction, Learning, and Games}, year = {2006}, isbn = {0521841089}, publisher = {Cambridge University Press}, address = {USA} 
}

@article{chopra2013effect,
  title={The effect of errors in means, variances, and covariances on optimal portfolio choice},
  author={Chopra, Vijay K. and Ziemba, William T.},
  journal={The Journal of Portfolio Management},
  volume={19},
  issue={2},
  pages={6--11},
  month={Winter},
  year={1993},
  doi={10.3905/jpm.1993.409440}
}

@inproceedings{duchi2008efficient,
author = {Duchi, John and Shalev-Shwartz, Shai and Singer, Yoram and Chandra, Tushar},
title = {Efficient projections onto the $l_1$-ball for learning in high dimensions},
year = {2008},
isbn = {9781605582054},
publisher = {Association for Computing Machinery},
address = {New York, NY, USA},
doi = {10.1145/1390156.1390191},
booktitle = {Proceedings of the 25th International Conference on Machine Learning},
pages = {272–279},
numpages = {8},
location = {Helsinki, Finland},
series = {ICML '08}
}

@article{clarkson2012sublinear,
  title={Sublinear optimization for machine learning},
  author={Clarkson, Kenneth L. and Hazan, Elad and Woodruff, David P.},
  journal={Journal of the ACM (JACM)},
  volume={59},
  number={5},
  pages={1--49},
  year={2012},
  publisher={ACM New York, NY, USA},
  doi={10.1145/2371656.2371658}
}

@inproceedings{Flaxman2004,
author = {Flaxman, Abraham D. and Kalai, Adam Tauman and McMahan, H. Brendan},
title = {Online convex optimization in the bandit setting: gradient descent without a gradient},
year = {2005},
isbn = {0898715857},
publisher = {Society for Industrial and Applied Mathematics},
address = {USA},
booktitle = {Proceedings of the Sixteenth Annual ACM-SIAM Symposium on Discrete Algorithms},
pages = {385–394},
numpages = {10},
location = {Vancouver, British Columbia},
series = {SODA '05}
}

@inproceedings{Fotakis2020,
 author = {Fotakis, Dimitris and Lianeas, Thanasis and Piliouras, Georgios and Skoulakis, Stratis},
 booktitle = {Advances in Neural Information Processing Systems},
 editor = {H. Larochelle and M. Ranzato and R. Hadsell and M.F. Balcan and H. Lin},
 pages = {7816--7827},
 publisher = {Curran Associates, Inc.},
 title = {Efficient Online Learning of Optimal Rankings: Dimensionality Reduction via Gradient Descent},
 volume = {33},
 year = {2020}, 
 address = {Virtual}
}

@article{Freund1997,
title = {A Decision-Theoretic Generalization of On-Line Learning and an Application to Boosting},
journal = {Journal of Computer and System Sciences},
volume = {55},
number = {1},
pages = {119-139},
year = {1997},
issn = {0022-0000},
doi = {https://doi.org/10.1006/jcss.1997.1504},
author = {Yoav Freund and Robert E Schapire}
}

@article{giovannetti2008architectures,
  title = {Architectures for a quantum random access memory},
  author = {Giovannetti, Vittorio and Lloyd, Seth and Maccone, Lorenzo},
  journal = {Phys. Rev. A},
  volume = {78},
  issue = {5},
  pages = {052310},
  numpages = {9},
  year = {2008},
  month = {Nov},
  publisher = {American Physical Society},
  doi = {10.1103/PhysRevA.78.052310}
}

@article{giovannetti2008quantum,
  title = {Quantum Random Access Memory},
  author = {Giovannetti, Vittorio and Lloyd, Seth and Maccone, Lorenzo},
  journal = {Phys. Rev. Lett.},
  volume = {100},
  issue = {16},
  pages = {160501},
  numpages = {4},
  year = {2008},
  month = {Apr},
  publisher = {American Physical Society},
  doi = {10.1103/PhysRevLett.100.160501}
}

@article{goldfarb2003robust,
author = {Goldfarb, D. and Iyengar, G.},
title = {Robust Portfolio Selection Problems},
journal = {Mathematics of Operations Research},
volume = {28},
number = {1},
pages = {1-38},
year = {2003},
doi = {10.1287/moor.28.1.1.14260}
}

@article{Grigoriadis1995,
    title = {{A sublinear-time randomized approximation algorithm for matrix games}},
    year = {1995},
    journal = {Operations Research Letters},
    author = {Grigoriadis, Michael D. and Khachiyan, Leonid G.},
    number = {2},
    pages = {53--58},
    volume = {18},
    doi = {10.1016/0167-6377(95)00032-0},
    issn = {01676377}
}

@article{hamoudi2022preparing,
  title = {Preparing many copies of a quantum state in the black-box model},
  author = {Hamoudi, Yassine},
  journal = {Phys. Rev. A},
  volume = {105},
  issue = {6},
  pages = {062440},
  numpages = {5},
  year = {2022},
  month = {Jun},
  publisher = {American Physical Society},
  doi = {10.1103/PhysRevA.105.062440}
}

@article{Hamoudi2019,
    title = {{Quantum and classical algorithms for approximate submodular function minimization}},
    year = {2019},
    journal = {Quantum Information and Computation},
    author = {Hamoudi, Yassine and Rebentrost, Patrick and Rosmanis, Ansis and Santha, Miklos},
    number = {15-16},
    pages = {1325--1349},
    volume = {19},
    doi = {10.26421/qic19.15-16-5},
    issn = {15337146},
    arxivId = {1907.05378}
}

@article{HarryMarkowitz1952,
 ISSN = {00221082, 15406261},
 author = {Harry Markowitz},
 journal = {The Journal of Finance},
 number = {1},
 pages = {77--91},
 publisher = {[American Finance Association, Wiley]},
 title = {Portfolio Selection},
 urldate = {2024-11-14},
 volume = {7},
 year = {1952},
 doi={10.2307/2975974}
}

@article{Hazan2016b,
    title = {{Introduction to Online Convex Optimization}},
    year = {2016},
    journal = {Foundations and Trends{\textregistered} in Optimization},
    author = {Hazan, Elad},
    number = {3-4},
    pages = {157--325},
    volume = {2},
    doi = {10.1561/2400000013},
    issn = {2167-3888},
    arxivId = {1909.05207}
}

@Article{hazan2007logarithmic,
author={Hazan, Elad
and Agarwal, Amit
and Kale, Satyen},
title={Logarithmic regret algorithms for online convex optimization},
journal={Machine Learning},
year={2007},
month={Dec},
day={01},
volume={69},
number={2},
pages={169-192},
issn={1573-0565},
doi={10.1007/s10994-007-5016-8}
}

@article{Hazan2016a,
    title = {{Learning rotations with little regret}},
    year = {2016},
    journal = {Machine Learning},
    author = {Hazan, Elad and Kale, Satyen and Warmuth, Manfred K.},
    number = {1},
    pages = {129--148},
    volume = {104},
    publisher = {Springer US},
    doi = {10.1007/s10994-016-5548-x},
    issn = {15730565}
}

@InProceedings{Hazan2007,
author="Helmbold, David P.
and Warmuth, Manfred K.",
editor="Bshouty, Nader H.
and Gentile, Claudio",
title="Learning Permutations with Exponential Weights",
booktitle="Learning Theory",
year="2007",
publisher="Springer Berlin Heidelberg",
address="Berlin, Heidelberg",
pages="469--483",
isbn="978-3-540-72927-3",
doi="10.1007/978-3-540-72927-3_34"
}

@article{Helmbold1998,
    title = {{On-line portfolio selection using multiplicative updates}},
    year = {1998},
    journal = {Mathematical Finance},
    author = {Helmbold, David P. and Schapire, Robert E. and Singer, Yoram and Warmuth, Manfred K.},
    number = {4},
    pages = {325--347},
    volume = {8},
    doi = {10.1111/1467-9965.00058},
    issn = {09601627}
}

@article{Helmbold2009,
author = {Helmbold, David P. and Warmuth, Manfred K.},
title = {Learning Permutations with Exponential Weights},
year = {2009},
issue_date = {12/1/2009},
publisher = {JMLR.org},
volume = {10},
issn = {1532-4435},
journal = {J. Mach. Learn. Res.},
month = {dec},
pages = {1705–1736},
numpages = {32}
}

@article{babbush2018encoding,
  title = {Encoding Electronic Spectra in Quantum Circuits with Linear T Complexity},
  author = {Babbush, Ryan and Gidney, Craig and Berry, Dominic W. and Wiebe, Nathan and McClean, Jarrod and Paler, Alexandru and Fowler, Austin and Neven, Hartmut},
  journal = {Phys. Rev. X},
  volume = {8},
  issue = {4},
  pages = {041015},
  numpages = {36},
  year = {2018},
  month = {Oct},
  publisher = {American Physical Society},
  doi = {10.1103/PhysRevX.8.041015},
  url = {https://link.aps.org/doi/10.1103/PhysRevX.8.041015}
}

@InProceedings{izdebski2020improved,
  author =	{Izdebski, Adam and {de Wolf}, Ronald},
  title =	{{Improved Quantum Boosting}},
  booktitle =	{31st Annual European Symposium on Algorithms (ESA 2023)},
  pages =	{64:1--64:16},
  series =	{Leibniz International Proceedings in Informatics (LIPIcs)},
  ISBN =	{978-3-95977-295-2},
  ISSN =	{1868-8969},
  year =	{2023},
  volume =	{274},
  editor =	{G{\o}rtz, Inge Li and Farach-Colton, Martin and Puglisi, Simon J. and Herman, Grzegorz},
  publisher =	{Schloss Dagstuhl -- Leibniz-Zentrum f{\"u}r Informatik},
  address =	{Dagstuhl, Germany},
  URN =		{urn:nbn:de:0030-drops-187178},
  doi =		{10.4230/LIPIcs.ESA.2023.64}
}

@ARTICLE{kim2016online,
  author={Kim, Seung-Jun and Giannakis, Geogios B.},
  journal={IEEE Transactions on Smart Grid}, 
  title={An Online Convex Optimization Approach to Real-Time Energy Pricing for Demand Response}, 
  year={2017},
  volume={8},
  number={6},
  pages={2784-2793},
  doi={10.1109/TSG.2016.2539948}}

@article{li2014online,
author = {Li, Bin and Hoi, Steven C. H.},
title = {Online portfolio selection: A survey},
year = {2014},
issue_date = {January 2014},
publisher = {Association for Computing Machinery},
address = {New York, NY, USA},
volume = {46},
number = {3},
issn = {0360-0300},
doi = {10.1145/2512962},
journal = {ACM Comput. Surv.},
month = {jan},
articleno = {35},
numpages = {36}
}

@Article{Mukhopadhyay2025quantum,
author={Mukhopadhyay, Priyanka},
title={A quantum random access memory ({QRAM}) using a polynomial encoding of binary strings},
journal={Scientific Reports},
year={2025},
month={Mar},
day={31},
volume={15},
number={1},
pages={11002},
issn={2045-2322},
doi={10.1038/s41598-025-95283-5},
url={https://doi.org/10.1038/s41598-025-95283-5}
}

@article{Arunachalam2015robustness,
doi = {10.1088/1367-2630/17/12/123010},
url = {https://dx.doi.org/10.1088/1367-2630/17/12/123010},
year = {2015},
month = {dec},
publisher = {IOP Publishing},
volume = {17},
number = {12},
pages = {123010},
author = {Arunachalam, Srinivasan and Gheorghiu, Vlad and Jochym-O’Connor, Tomas and Mosca, Michele and Srinivasan, Priyaa Varshinee},
title = {On the robustness of bucket brigade quantum RAM},
journal = {New Journal of Physics}
}

@inProceedings{li2019sublinear,
  title =  {Sublinear quantum algorithms for training linear and kernel-based classifiers}, 
  author =       {Li, Tongyang and Chakrabarti, Shouvanik and Wu, Xiaodi},  
  year =  {2019}, 
  month =  {09--15 Jun}, 
  booktitle =  {Proceedings of the 36th International Conference on Machine Learning},  
   publisher =    {PMLR}, 
  address   = {Long Beach, California, USA}, 
  series =  {Proceedings of Machine Learning Research},  
  volume =  {97},  series =  {Proceedings of Machine Learning Research}, 
  pages =  {3815--3824},  
  editor =  {Chaudhuri, Kamalika and Salakhutdinov, Ruslan}, url =  {https://proceedings.mlr.press/v97/li19b.html}
}

@article{liu2017survey,
  title={Survey of convex optimization for aerospace applications},
  author={Liu, Xinfu and Lu, Ping and Pan, Binfeng},
  journal={Astrodynamics},
  volume={1},
  pages={23--40},
  year={2017},
  publisher={Springer},
  doi={10.1007/s42064-017-0003-8}
}

@article{bellante2025quantum,
  title={Quantum Sparse Recovery and Quantum Orthogonal Matching Pursuit},
  author={Bellante, Armando and Vanerio, Stefano and Zanero, Stefano},
  journal={arXiv preprint arXiv:2510.06925},
  year={2025}
}

@article{chen2024low,
  title={On the Low-Temperature MCMC threshold: the cases of sparse tensor {PCA}, sparse regression, and a geometric rule},
  author={Chen, Zongchen and Sheehan, Conor and Zadik, Ilias},
  journal={arXiv preprint arXiv:2408.00746},
  year={2024}
}

@InProceedings{bertsimas2004robust,
author="Bertsimas, Dimitris
and Thiele, Aur{\'e}lie",
editor="Bienstock, Daniel
and Nemhauser, George",
title="A Robust Optimization Approach to Supply Chain Management",
booktitle="Integer Programming and Combinatorial Optimization",
year="2004",
publisher="Springer Berlin Heidelberg",
address="Berlin, Heidelberg",
pages="86--100",
isbn="978-3-540-25960-2"
}

@Article{Korn2022,
author={Korn, Ralf
and M{\"u}ller, Lukas},
title={Optimal portfolios in the presence of stress scenarios A worst-case approach},
journal={Mathematics and Financial Economics},
year={2022},
month={Jan},
day={01},
volume={16},
number={1},
pages={153-185},
issn={1862-9660},
doi={10.1007/s11579-021-00304-2},
url={https://doi.org/10.1007/s11579-021-00304-2}
}

@article{luo2025space,
  title={Space-time tradeoff for sparse quantum state preparation},
  author={Luo, Jingquan and Li, Guanzhong and Li, Lvzhou},
  journal={arXiv preprint arXiv:2506.16964},
  year={2025}
}

@misc{lobo2000worst,
  title={The worst-case risk of a portfolio},
  author={Lobo, Miguel Sousa and Boyd, Stephen},
  year={2000}
}

@inproceedings{mcmahan2010adaptive,
  author       = {H. Brendan McMahan and
                  Matthew J. Streeter},
  title        = {Adaptive Bound Optimization for Online Convex Optimization},
  booktitle    = {{COLT} 2010 - The 23rd Conference on Learning Theory},
  pages        = {244--256},
  publisher    = {Omnipress},
  year         = {2010}, 
  address      =  {Haifa, Israel,}
}

@inproceedings{nayak1999quantum,
author = {Nayak, Ashwin and Wu, Felix},
title = {The quantum query complexity of approximating the median and related statistics},
year = {1999},
isbn = {1581130678},
publisher = {Association for Computing Machinery},
address = {New York, NY, USA},
doi = {10.1145/301250.301349},
booktitle = {Proceedings of the Thirty-First Annual ACM Symposium on Theory of Computing},
pages = {384–393},
numpages = {10},
location = {Atlanta, Georgia, USA},
series = {STOC '99}
}

@inproceedings{berry2014exponential,
author = {Berry, Dominic W. and Childs, Andrew M. and Cleve, Richard and Kothari, Robin and Somma, Rolando D.},
title = {Exponential improvement in precision for simulating sparse {H}amiltonians},
year = {2014},
isbn = {9781450327107},
publisher = {Association for Computing Machinery},
address = {New York, NY, USA},
doi = {10.1145/2591796.2591854},
booktitle = {Proceedings of the Forty-Sixth Annual ACM Symposium on Theory of Computing},
pages = {283–292},
numpages = {10},
location = {New York, New York},
series = {STOC '14}
}

@inproceedings{gilyen2019quantum,
author = {Gily\'{e}n, Andr\'{a}s and Su, Yuan and Low, Guang Hao and Wiebe, Nathan},
title = {Quantum singular value transformation and beyond: exponential improvements for quantum matrix arithmetics},
year = {2019},
isbn = {9781450367059},
publisher = {Association for Computing Machinery},
address = {New York, NY, USA},
doi = {10.1145/3313276.3316366},
booktitle = {Proceedings of the 51st Annual ACM SIGACT Symposium on Theory of Computing},
pages = {193–204},
numpages = {12},
location = {Phoenix, AZ, USA},
series = {STOC 2019}
}

@phdthesis{prakash2014quantum,
  title={Quantum algorithms for linear algebra and machine learning},
  author={Prakash, Anupam},
  year={2014},
  school ={University of California, Berkeley},
  url={https://escholarship.org/uc/item/5v9535q4}
}

@article{harrow2009quantum,
  title = {Quantum Algorithm for Linear Systems of Equations},
  author = {Harrow, Aram W. and Hassidim, Avinatan and Lloyd, Seth},
  journal = {Phys. Rev. Lett.},
  volume = {103},
  issue = {15},
  pages = {150502},
  numpages = {4},
  year = {2009},
  month = {Oct},
  publisher = {American Physical Society},
  doi = {10.1103/PhysRevLett.103.150502}
}

@Article{berry2007efficient,
author={Berry, Dominic W.
and Ahokas, Graeme
and Cleve, Richard
and Sanders, Barry C.},
title={Efficient Quantum Algorithms for Simulating Sparse {H}amiltonians},
journal={Communications in Mathematical Physics},
year={2007},
month={Mar},
day={01},
volume={270},
number={2},
pages={359-371},
issn={1432-0916},
doi={10.1007/s00220-006-0150-x},
url={https://doi.org/10.1007/s00220-006-0150-x}
}

@book{nielsen2010quantum,
  title={Quantum computation and quantum information},
  author={Nielsen, Michael A and Chuang, Isaac L},
  year={2010},
  publisher={Cambridge university press},
  address={New Yok}, 
  doi={10.1017/CBO9780511976667}
}

@article{Doriguello2025quantumalgorithms,
  doi = {10.22331/q-2025-03-25-1674},
  url = {https://doi.org/10.22331/q-2025-03-25-1674},
  title = {Quantum {A}lgorithms for the {P}athwise {L}asso},
  author = {Doriguello, Joao F. and Lim, Debbie and Pun, Chi Seng and Rebentrost, Patrick and Vaidya, Tushar},
  journal = {{Quantum}},
  issn = {2521-327X},
  publisher = {{Verein zur F{\"{o}}rderung des Open Access Publizierens in den Quantenwissenschaften}},
  volume = {9},
  pages = {1674},
  month = mar,
  year = {2025}
}

@article{patel2011trajectory,
  title={Trajectory generation for aircraft avoidance maneuvers using online optimization},
  author={Patel, Rushen B. and Goulart, Paul J.},
  journal={Journal of guidance, control, and dynamics},
  volume={34},
  number={1},
  pages={218--230},
  year={2011},
  doi={10.2514/1.49518}
}

@InProceedings{Schraudolph2007,
  title = 	 {A Stochastic Quasi-{N}ewton Method for Online Convex Optimization},
  author = 	 {Schraudolph, Nicol N. and Yu, Jin and Günter, Simon},
  booktitle = 	 {Proceedings of the Eleventh International Conference on Artificial Intelligence and Statistics},
  pages = 	 {436--443},
  year = 	 {2007},
  editor = 	 {Meila, Marina and Shen, Xiaotong},
  volume = 	 {2},
  series = 	 {Proceedings of Machine Learning Research},
  address = 	 {San Juan, Puerto Rico},
  month = 	 {21--24 Mar},
  publisher =    {PMLR}
}

@InProceedings{servedio2003smooth,
author="Servedio, Rocco A.",
editor="Helmbold, David
and Williamson, Bob",
title="Smooth Boosting and Learning with Malicious Noise",
booktitle="Computational Learning Theory",
year="2001",
publisher="Springer Berlin Heidelberg",
address="Berlin, Heidelberg",
pages="473--489",
isbn="978-3-540-44581-4",
doi="10.1007/3-540-44581-1_31"
}

@article{shalev2012online,
  title={Online learning and online convex optimization},
  author={Shalev-Shwartz, Shai},
  journal={Foundations and Trends{\textregistered} in Machine Learning},
  volume={4},
  number={2},
  pages={107--194},
  year={2012},
  publisher={Now Publishers, Inc.},
  doi={10.1561/2200000018}
}

@article{Vose1991,
    title = {A Linear Algorithms for Generating Random Numbers with A Given Distribution},
    year = {1991},
    journal = {IEEE Transactions on Software Engineering},
    author = {Vose, M.D.},
    number = {972},
    volume = {17},
    isbn = {9788490225370},
    doi = {10.1109/32.92917}
}

@article{Wang2014,
    title = {{Randomized block coordinate descent for online and stochastic optimization}},
    year = {2014},
    journal = {arXiv preprint arXiv:1407.0107},
    author = {Wang, Huahua and Banerjee, Arindam},
    doi = {10.48550/arXiv.1407.0107}
}

@inproceedings{zinkevich2003online,
author = {Zinkevich, Martin},
title = {Online convex programming and generalized infinitesimal gradient ascent},
year = {2003},
isbn = {1577351894},
publisher = {AAAI Press},
booktitle = {Proceedings of the Twentieth International Conference on International Conference on Machine Learning},
pages = {928–935},
numpages = {8},
address = {Washington, DC, USA},
series = {ICML'03}
}

\end{document}